\documentclass[12pt,authoryear,review]{elsarticle}
\usepackage[latin9]{inputenc}
\usepackage[a4paper]{geometry}
\usepackage{verbatim}
\usepackage{url}
\usepackage{amsmath}
\usepackage{mathrsfs}
\usepackage{amsthm}
\usepackage{amssymb}
\usepackage{scalefnt}
\usepackage{amsfonts}
\usepackage[singlespacing]{setspace}
\usepackage[bottom]{footmisc}
\usepackage{indentfirst}
\usepackage{endnotes}
\usepackage{rotating}
\usepackage{pgf}
\usepackage{epstopdf}
\usepackage{thmtools}
\usepackage{float}
\usepackage{mathpazo}
\usepackage{chngcntr}
\usepackage{apptools}
\AtAppendix{\counterwithin{lemma}{section}}

\usepackage[unicode=true,
 bookmarks=false,
 breaklinks=false,pdfborder={0 0 1}, colorlinks=false]
 {hyperref}
\hypersetup{
 colorlinks=true, linkcolor= blue, citecolor=blue, urlcolor=blue}

\theoremstyle{plain}
\newtheorem{theorem}{Theorem}

\newtheorem{corollary}{Corollary}
\newtheorem{assumption}{Assumption}
\newtheorem{definition}{Definition}
\newtheorem{example}{Example}
\newtheorem{lemma}{Lemma}
\newtheorem{proposition}{Proposition}
\theoremstyle{remark}
\DeclareMathOperator*{\argmax}{arg\,max}
\DeclareMathOperator*{\argmin}{arg\,min}

\newcommand\cites[1]{\citeauthor{#1}'s\ (\citeyear{#1})}
\newcommand{\com}[1]{}

\begin{document}

\renewcommand\thmcontinues[1]{Continued}

\begin{frontmatter}{}

\title{Auction design with ambiguity:\\
Optimality of the first-price and all-pay auctions}

\author[A1]{Sosung Baik}

\ead{sosung.baik@gmail.com}

\author[A1]{Sung-Ha Hwang\corref{cor1}}

\ead{sungha@kaist.ac.kr}

\cortext[cor1]{\today. Corresponding author. The research of S.-H. H. was supported by
the National Research Foundation of Korea.}

\address[A1]{Korea Advanced Institute of Science and Technology (KAIST), Seoul, Korea}

\begin{abstract}

We study the optimal auction design problem when bidders' preferences follow the maxmin expected utility model. We suppose that each bidder's set of priors consists of beliefs close to the seller's belief, where  ``closeness'' is defined by a divergence. For a given allocation rule, we identify a class of optimal transfer candidates, named the \textit{win-lose dependent transfers}, with the following property: each type of bidder's transfer conditional on winning or losing is independent of the competitor's type report. Our result reduces the infinite-dimensional optimal transfer problem to a two-dimensional optimization problem. By solving the reduced problem, we find that: (i) among efficient mechanisms with no premiums for losers, the first-price auction is optimal; and, (ii) among efficient winner-favored mechanisms where each bidder pays smaller amounts when she wins than loses: the all-pay auction is optimal. Under a simplifying assumption, these two auctions remain optimal under the endogenous allocation rule.

\end{abstract}
\end{frontmatter}

\thispagestyle{empty}

\newpage

\section{Introduction\setcounter{page}{1}}

We study the optimal auction design problem with ambiguity where bidders' preferences follow the maxmin expected utility (MMEU) model \citep{Gil89}. In the MMEU model, the decision maker holds multiple priors on the state and maximizes the \textit{worst-case utility}, the minimum expected utility over the set of priors. Despite a growing interest in auctions with ambiguity\footnote{\citet{Lo98, Bo06, Bo09, Bodoh12, Chi15, Lao19, Aus20, Kocyi20, Ghosh21}; see Section \ref{sec:literature} for details.}, existing results in the popular setup of independent private values (IPV) auctions are somewhat negative. The optimal mechanism turns out to be a full insurance auction, where the seller provides premiums to losers so that each bidder's payoff remains constant with respect to the competitor's type report \citep{Bo06}. However, the full insurance auction is rarely observed in reality. This gap between theory and practice poses a puzzle to researchers. In addition, the revenue ranking between the first-price and second-price auctions is sensitive to the choice of the bidders' sets of priors \citep{Bodoh12}.

In this paper, we ask a question rather different from the existing literature: under which circumstances does one of the standard auction formats, such as the first-price, second-price or all-pay auctions, become optimal? Our interest in these mechanisms is motivated by their being easily implementable since they do not require the designer's specific knowledge, for example, the bidders' beliefs about each other's types \citep{Wil87}.\footnote{Unlike these standard formats, the implementation of the full insurance auction requires the seller to know the bidders' sets of priors, which depends on the specific environment.}
We find that under  plausible assumptions on the set of priors, (i) the first-price auction is optimal among efficient mechanisms in which the seller provides no premiums; and, (ii) the all-pay auction is optimal among efficient mechanisms in which each bidder pays smaller amounts when she wins than loses.

To derive our results, we first suppose that each bidder's set of priors consists of beliefs close to some reference belief, interpreted as either the true probability law or the focal point distribution \citep{Lo98, Bo06}. We measure ``closeness'' between beliefs using a \textit{divergence}---a coefficient of discrepancy between two probability distributions used in statistics and information theory \citep{Kull51, Kull59, Ali66}. Our assumption on the set of priors is general enough to incorporate interesting examples in the existing literature (Example \ref{ex:div}), yet specific enough to produce a definite answer to the optimal mechanism problem (Theorem \ref{thm:main} and Propositions \ref{prop:limited_prem}-\ref{prop:winner_favored}).

Then, we study the problem of finding the revenue-maximizing transfer rule for a given allocation rule, called the \textit{optimal transfer problem}. Besides the usual feasibility constraint (i.e., incentive compatibility and individual rationality), to identify environments where some standard format becomes optimal, we assume that the seller faces an additional constraint (Assumption \ref{assum:T}). Since this additional constraint can be chosen to be trivially satisfied, our setup also incorporates the standard setup with only the feasibility constraint.

Our first contribution lies in identifying a class of optimal transfer candidates (Theorem \ref{thm:main}) with the following special properties: (i) each bidder's transfer (and hence the payoff) conditional on winning or losing is independent of the competitor's type report, and (ii) her winning payoff is higher than the losing payoff. We name the mechanism (or transfer) with this property as the \textit{win-lose dependent} mechanism (or transfer) (Definition \ref{def:win_lose_dependent}). Hence, the win-lose dependent mechanism generalizes the full insurance mechanism by allowing the winning and losing payoffs to differ. Moreover, for win-lose dependent mechanisms, the feasibility constraint simplifies to a formula for the bidder's interim worst-case utility analogous to the standard envelope characterization in Bayesian mechanism design \citep{Myer81}. These results reduce the infinite-dimensional problem of choosing the transfer rule to a two-dimensional problem of determining each type of bidder's winning and losing payoffs.

\com{
Theorem \ref{thm:main} can be regarded as a generalization of \cites{Bo06} result that the full insurance transfer with the same interim worst case utility as a given feasible transfer generates a higher revenue than the original. We generalize this result by showing the following principle: to an ambiguity averse decision maker, the conditional expectation of a payoff schedule (with respect to coarser information) yields a higher worst-case utility than the original (Proposition \ref{prop:SOSD}). In other words, the decision maker prefers payoff schedules with less variation.\footnote{More precisely, the MMEU preference exhibits monotonicity with respect to the second-order stochastic dominance.} Thus, given any transfer, by offering its conditional expectation with respect to the bidder's winning and losing events, the seller can create room to extract greater revenues. As a special case, the payoff with no variation (or a constant payoff function as in the full insurance auction) can yield a higher revenue than the original. The novel insight from our approaches, then, lies in showing that seemingly unrelated auctions including full insurance auctions, first-price auctions, and all-pay auctions belong to a single class of transfers in which transfer functions are piecewise constant depending on winning and losing and transfers in this class provide higher revenues than those not in this class such as the second price auction and the war of attrition.
}

Theorem \ref{thm:main} can be regarded as a generalization of \cites{Bo06} result on the optimality of the full insurance mechanisms. To prove Theorem \ref{thm:main}, we establish the following principle: to an ambiguity averse decision maker, the conditional expectation of a payoff schedule (with respect to coarser information) yields a higher worst-case utility than the original (Proposition \ref{prop:SOSD}).\footnote{In other words, the MMEU preference exhibits monotonicity with respect to (a special kind of) the second-order stochastic dominance.} Thus, given any transfer, by offering its conditional expectation with respect to the bidders' winning and losing events, the seller can create room to extract greater revenues. Exploiting this property along with characterizations of feasible mechanisms (Proposition \ref{prop:feasibility}), for any feasible transfer, we construct a win-lose dependent transfer which yields a greater revenue than the original while preserving incentive compatibility. Our result provides a novel insight that seemingly unrelated formats---the full insurance, first-price and all-pay auctions---share the common property of being win-lose dependent, which makes them superior over other formats such as the second-price auction and the war of attrition.

The second contribution of our paper is that using Theorem \ref{thm:main}, we find two plausible classes of transfers within which the first-price and all-pay auctions become optimal, respectively: (i) the limited premium transfers (Section \ref{subsec:limited_prem}), and (ii) the winner-favored transfers (Section \ref{subsec:winner_favored}). Under \textit{limited premium transfers}, the maximum premium provided to losers is limited by a certain amount. This class describes situations in which the seller provides only partial premiums to losers, as in real-world premium auctions such as Amsterdam auctions;\footnote{\citet{Mil04, Goe04, Hu11, Brun14, Hu18}.} see the related literature in Section \ref{sec:literature}.

Under \textit{winner-favored transfers}, each bidder's winning transfer is less than (or equal to) the losing transfer. This class interests us because it contains auctions whose revenues are frequently compared in the contest literature---the all-pay auction vs. the war of attrition \citep{Kris97, Horisch10, Bos2012}, and the all pay-auction vs. the sad loser auction \citep{Ril81, Minchuk18}.\footnote{For instance, \citet{Kris97} show that the war of attrition revenue dominates the all-pay auction in the affiliated interdependent values environment \citep{Mil82}. Also, \citet{Minchuk18} compares the all-pay and sad loser auctions with non-linear effort functions, and finds that the sad loser auction revenue dominates the all-pay auction when the cost of effort is convex, and the opposite holds when the cost of effort is concave.} In all three auctions, the loser pays exactly her bid (possibly because bids are sunk). On the other hand, the winner pays her bid in the all-pay auction, the second-highest bid in the war of attrition, and nothing in the sad loser auction.  Thus, the relationship between the all-pay auction and the war of attrition is analogous to that between the first-price and second-price auctions. Also, the sad loser auction is obtained from the all-pay auction via the full reimbursement of the winner's bid. In contests, the winner's cost of efforts (bids) are often reimbursed to encourage contestants' higher efforts;\footnote{\citet{Kap02, Coh08, Matros09, Matros12, Min20}.} see the related literature in Section \ref{sec:literature}.

Within the first class of limited premium transfers, our results are as follows (Proposition \ref{prop:limited_prem} and Corollary \ref{cor:FPA_FIA}; Figures \ref{fig:solution_w_l_limited_prem} and \ref{fig:tw_tl_graph}). When the seller cannot provide premiums, the first-price auction is optimal among efficient mechanisms. This class includes many standard formats, for instance, the second-price and all-pay auctions. In contrast, when the seller can provide sufficiently large premiums, a full insurance auction is optimal, reproducing \cites{Bo06} result. When the seller can provide some premiums but cannot fully insure all types, the optimal mechanism is a hybrid of the two mechanisms in which low types are fully insured, whereas high types are only partially insured.

Within the second class of winner-favored mechanisms, the all-pay auction is optimal among efficient mechanisms (Proposition \ref{prop:winner_favored} and Corollary \ref{cor:APA}; Figure \ref{fig:solution_w_l_APA}). This result implies that within the class of  auctions where the loser pays her own bid, the all pay auction revenue dominates any auctions (i) in which the winner pays any amount between the second-highest bid and her own bid, and (ii) in which the winner is partially or fully reimbursed for her bid. 

The results explained thus far focus on finding the optimal transfer rule, taking the allocation rule as given. In Section \ref{sec:opt_alloc_rule}, under a simplifying assumption on the divergence (Example \ref{ex:div} (b)), we endogenously determine the optimal allocation rule. With endogenous allocation, the first-price and all-pay auctions with suitable reserve prices remain optimal in the classes of no premium and winner-favored mechanisms (Corollary \ref{cor:optimal_mechanism_endogenous_x}). This suggests that the assumption of the exogenous allocation rule may be unessential for our results, and that the two formats have superior performances in a wide range of environments.

This paper is organized as follows. Section \ref{sec:model} introduces the setup. Section \ref{sec:main} presents and proves the optimality of the win-lose dependent transfers. In Section \ref{sec:constraints_applications}, we solve the optimal transfer problem within the classes of limited premium transfers and winner-favored transfers. Section \ref{sec:opt_alloc_rule} studies the endogenous optimal allocation rule. Section \ref{sec:diss} discusses the related literature and concludes.

\section{Model} \label{sec:model}

\subsection{Agents and preferences}

We follow the standard setup in the literature on mechanism design with ambiguity averse agents \citep{Bo06, Bodoh12, Wol16, Kocher19}. A seller wants to sell an indivisible object to two bidders. The assumption of two bidders is unessential; our results readily extend to the model with an arbitrary number of bidders. Each bidder has a privately known type $\theta \in \Theta = [\underline \theta, \bar \theta] \subset \mathbb{R}$ representing her valuation for the object, where $0 < \underline \theta < \bar \theta$ and $\Theta$ is equipped with the Borel $\sigma$-algebra $\mathcal{B}$. The bidders' types are drawn independently.

Each bidder, unsure of the exact type distribution, holds a \textit{set of priors} $\Delta \subset \Delta(\Theta, \mathcal{B})$ rather than a single prior about the competitor's type, where $\Delta(\Theta, \mathcal{A})$ denotes the set of probability measures on a $\sigma$-algebra $\mathcal{A}$. Following the MMEU model \citep{Gil89}, each bidder evaluates a payoff schedule\footnote{Throughout the paper, a payoff schedule is a bounded measurable function $\pi: \Theta \rightarrow \mathbb{R}$; we interpret $\pi(\theta)$ as the payoff when state $\theta$ is realized.} (or often called an \textit{act} in the literature) by its \textit{worst-case utility}, defined as the minimum expected payoff over the set of priors.  A belief that minimizes the expected payoff is called a \textit{worst-case belief}.

The seller is ambiguity neutral, and hence has a single prior $P \in \Delta(\Theta, \mathcal{B})$ over each bidder's type, assumed to be atomless. We call $P$ the \textit{reference belief}; this can be interpreted as either the true probability distribution of types \citep{Lo98} or a focal point around which bidders' beliefs are perturbed \citep{Bo06}. We discuss the ambiguity averse seller in the concluding section.

Each bidder's set of priors $\Delta$ consists of probability measures close to the reference belief $P$, where the ``closeness'' between two beliefs is measured by a \textit{divergence}. To explain the concept of divergence, we clarify some terminologies. By a $\sigma$-algebra, we always mean a sub-$\sigma$-algebra of the Borel $\sigma$-algebra $\mathcal{B}$. Also, whenever we say ``almost everywhere'' (abbreviated as a.e.), we mean ``almost everywhere with respect to $P$ (or the product measure $P^2$, depending on the context)''. Finally, given two $\sigma$-algebras $\mathcal{E} \subset \mathcal{A}$ and a probability measure $P \in \Delta(\Theta, \mathcal{A})$, the restriction of $P$ to $\mathcal{E}$ is denoted as $P_\mathcal{E} \in \Delta(\Theta, \mathcal{E})$.

A divergence $D(\cdot||\cdot)$ assigns a non-negative number (possibly infinite) to each pair of probability measures defined on the same $\sigma$-algebra.\footnote{Formally, $D: \bigcup_{\mathcal{A}} \Delta(\Theta, \mathcal{A}) \times \Delta(\Theta, \mathcal{A}) \rightarrow [0, \infty]$, with the union taken over all $\sigma$-algebras.} Given two probability measures $P$ and $Q$ on the same $\sigma$-algebra, we measure how far, or distinguishable, $Q$ is from $P$ by the divergence of $Q$ from $P$, $D(Q||P)$. Following the literature on statistics and information theory \citep{Kull51, Kull59, Ali66},\footnote{Properties \textbf{D4} and \textbf{D5} are related to \cites{Kull51} Theorem 4.1, \cites{Kull59} Corollary 3.2, and \cites{Ali66} second property.} we require that the divergence satisfies the following basic properties:

\begingroup
\renewcommand\theassumption{D}
\begin{assumption} [Divergence] \label{assum:div}

For every $\sigma$-algebra $\mathcal{A}$ and $P, Q \in \Delta(\Theta, \mathcal{A})$, we have:

\noindent \textbf{D1.} $D(Q||P) = 0$ if $Q = P$.

\noindent \textbf{D2.} If $Q \ll P$ and $\frac{dQ}{dP}$ is bounded, $D(\epsilon Q + (1-\epsilon)P || P)$ is continuous in $\epsilon \in [0, 1]$.

\noindent \textbf{D3.} If $D(Q||P) < \infty$, then $Q \ll P$.

\noindent \textbf{D4.} For a sub-$\sigma$-algebra $\mathcal{E}$ of $\mathcal{A}$, $D(Q_\mathcal{E}||P_\mathcal{E}) \leq D(Q||P)$.

\noindent \textbf{D5.} For a sub-$\sigma$-algebra $\mathcal{E}$ of $\mathcal{A}$, $D(Q_\mathcal{E}||P_\mathcal{E}) = D(Q||P)$ if $\frac{dQ_\mathcal{E}}{dP_\mathcal{E}} = \frac{dQ}{dP}$ a.e.

\end{assumption}
\endgroup

Properties \textbf{D1} (indistinguishability) and \textbf{D2} (continuity) are minimal requirements. \textbf{D3} says that if one can observe an event generated by $Q$ that never occurs under $P$ (i.e., $Q \not \ll P$), then $Q$ is perfectly distinguishable from $P$ (i.e., $D(Q||P) = \infty$). \textbf{D4} means that under coarser information (a sub-$\sigma$-algebra), one is less able to distinguish between two probability measures. \textbf{D5} means that if the likelihood ratio $\frac{dQ}{dP}$ is preserved under the coarser information, then the distinguishability between the two probability measures is also preserved. This is a natural requirement because the likelihood ratio $\frac{dQ}{dP}$ is a sufficient statistic for distinguishing between $P$ and $Q$ \citep[][Chapter 2]{Cox74}.

The bidder's set of priors is given as the divergence neighborhood of the reference belief $P$:
\[
\Delta = \{ Q \in \Delta(\Theta, \mathcal{B}): D(Q||P) \leq \eta\},
\]
where $\eta > 0$ is the degree of ambiguity. We provide two examples of divergence; in Appendix \ref{appen:model}, we verify that they indeed satisfy Assumption \ref{assum:div}.

\begin{example} \label{ex:div} \normalfont \,

\noindent \textbf{(a) The $\phi$-divergence and the relative entropy \citep{Han01}.}\\
For a convex continuous function $\phi: \mathbb{R}_+ \rightarrow \mathbb{R}$ with $\phi(1) = 0$, define the $\phi$-divergence as follows \citep{Ali66}:
\[
D(Q||P) := 	\int_\Theta \phi\left(\frac{dQ}{dP}\right) dP \,\,\, \text{for  $Q \ll P$}
\]
and $D(Q||P)=\infty$ otherwise. For example, when $\phi(x) \equiv (x-1)^2$, the $\phi$-divergence becomes the variance of the likelihood ratio $\frac{dQ}{dP}$. More generally, the $\phi$-divergence measures the dispersion of the likelihood ratio $\frac{dQ}{dP}$ by evaluating the expectation of a convex function of the likelihood ratio.

If $\phi(\alpha) \equiv \alpha \log \alpha$,\footnote{When $\alpha = 0$, we adopt the convention that $0 \log 0 = \lim_{\alpha \rightarrow 0+} \alpha \log \alpha = 0$. \label{fn:0log0}} the $\phi$-divergence becomes the popular \textit{relative entropy}, also known as the \textit{Kullback-Leibler divergence} \citep{Kull51}:
\[
D(Q||P) := \int_\Theta \frac{dQ}{dP} \log \frac{dQ}{dP} dP = \int_\Theta \log \frac{dQ}{dP} dQ \quad \text{for $Q \ll P$.}
\]
This divergence is used, among many others, in \cite{Han01}.

\medskip

\noindent \textbf{(b) The contamination model \citep{Bo06}.}  \\
Define the \textit{contamination divergence} as follows:
\[
D(Q||P) := \textnormal{ess\,sup}_P \left( 1 - \frac{dQ}{dP} \right) \quad \text{if $Q \ll P$},
\]
and $D(Q||P)=\infty$ otherwise. Intuitively, this divergence measures the maximum downward deviation of the likelihood ratio $\frac{dQ}{dP}$ from $1$. In Appendix \ref{appen:model}, we show that this divergence generates the set of priors known as the \textit{contamination model} (hence the name contamination divergence):
\[
\Delta = \{ Q \in \Delta(\Theta, \mathcal{B}): Q = (1-\eta)P + \eta R \text{ for some } R \in \Delta(\Theta, \mathcal{B}) \text{ with } R \ll P \}.
\]
This model has been used extensively in the literature on mechanism design with ambiguity \citep[e.g., ][]{Bo06, Bo09, Aus18}. $\square$
\end{example}

\subsection{Mechanism design}

The seller's problem is to design a mechanism which maximizes the expected revenue under the reference belief $P$. By the revelation principle,\footnote{Following most of the literature on mechanism design with ambiguity \citep{Bo06, Bodoh12, Wol16, Kocher19}, we focus on static mechanisms; in dynamic mechanisms, due to dynamic inconsistency of the MMEU model, the equilibrium outcome may not be replicable by a static direct mechanism \citep{Bo09}. In this setup, assuming that bidders cannot hedge against ambiguity by randomization, the proof of the revelation principle in our setup is standard; see \cite{Wol16} for details.} we can restrict our attention to direct mechanisms, defined as follows:


\begin{definition} \label{def:direct_mech}
A direct mechanism $(x, t)$ consists of two bounded measurable functions, namely an allocation rule $x = (x_1, x_2): \Theta^2 \rightarrow \mathbb{R}^2$ and a transfer rule $t = (t_1, t_2): \Theta^2 \rightarrow \mathbb{R}^2$. The allocation rule $x$ satisfies the following probability conditions: (i) $x_1(\theta, \theta') \geq 0$ and $x_2(\theta', \theta) \geq 0$, and (ii) $x_1(\theta, \theta') + x_2(\theta', \theta) \leq 1$ for $\theta, \theta' \in \Theta$.
\end{definition}

\noindent When bidder $i$ reports $\theta$ and her competitor reports $\theta'$, bidder $i$ wins the object with probability $x_i(\theta, \theta')$ and pays $t_i(\theta, \theta')$ to the seller.

Given a mechanism $(x, t)$, we introduce the following notation:
\begin{align}
X_i(\theta) := &\int_\Theta x_i(\theta, \theta') dP(\theta') \label{eq:notation_X} \\
X_i^{\min}(\theta) := &\inf_Q \{ \int_\Theta x_i(\theta, \theta') dQ(\theta'): D(Q||P) \leq \eta \} \label{eq:notation_Xmin} \\
U_i^{\min}(\theta) := &\inf_Q \{ \int_\Theta [ \theta x_i(\theta, \theta') - t_i(\theta, \theta') ] dQ(\theta'): D(Q||P) \leq \eta\} \label{eq:notation_Umin} \\
T_i(\theta) := &\int_\Theta t_i(\theta, \theta') dP(\theta') \label{eq:notation_T}.
\end{align}
We call $X_i(\theta)$ and $X_i^{\min}(\theta)$ the reference and the minimum winning probability of bidder $i$ with type $\theta$, respectively. Also, we call $U_i^{\min}(\theta)$ the bidder's interim worst-case utility and $T_i(\theta)$ the interim expected revenue. Likewise, the interim worst-case utility and revenue under $(x, \hat t)$ are written as $\widehat U_i^{\min}(\theta)$ and $\widehat T_i(\theta)$.

Since the bidders' preferences follow the MMEU model, the feasibility constraint (incentive compatibility and individual rationality) is as follows:
\begin{definition} \label{def:IC_IR}
(i) We say that $(x, t)$ is incentive compatible if, for every $i$ and $\theta$,
\[
   \theta \in \argmax_{\hat \theta \in \Theta} \,\, \inf_Q \{ \int_\Theta [ \theta x_i(\hat \theta, \theta') - t_i(\hat \theta, \theta') ] d Q(\theta') : D(Q||P) \leq \eta \}.
\]
(ii) We say that $(x, t)$ is individually rational if, for every $i$ and $\theta$, $U_i^{\min}(\theta) \geq 0$.\\
(iii) We say that $(x, t)$ is feasible if it is incentive compatible and individually rational.
\end{definition}

Our study mainly focuses on the problem of finding the optimal transfer rule, taking an allocation rule $x$ as exogenously given (Sections \ref{sec:main}-\ref{sec:constraints_applications}); the endogenous determination of the optimal allocation rule is studied in Section \ref{sec:opt_alloc_rule}. We impose the following regularity conditions on the allocation rule:

\begingroup
\renewcommand\theassumption{X}
\begin{assumption} \label{assum:X}

(i) For $\theta' \neq \theta$, we have $x_i(\theta, \theta') \in \{0, 1\}$.

\noindent (ii) If $X_i(\theta)=1$, then $\theta = \bar \theta$.

\noindent (iii) $X_i^{\min}(\theta)$ is non-decreasing in $\theta$.

\end{assumption}
\endgroup

\noindent Assumption \ref{assum:X} (i) means that the winner is chosen deterministically based on the reported type profile (except in the case of a tie). Assumption \ref{assum:X} (ii) means that no type of bidder wins with certainty except for the highest possible type. Assumption \ref{assum:X} (iii), analogously to the usual monotonicity condition \citep{Myer81}, ensures that $x$ is implementable.\footnote{However, unlike in \cite{Myer81}, Assumption \ref{assum:X} (iii) is only sufficient, not necessary, for implementability.}

As mentioned in the introduction, our aim is to identify environments where a commonly used auction format---such as the first-price, second-price, or all-pay auction---becomes optimal. To do this, we consider two classes of transfer rules: (i) the limited premium transfers under which the premium provided to losers is limited by a certain amount (Section \ref{subsec:limited_prem}), and (ii) the winner-favored transfers under which each bidder pays smaller amounts when she wins than loses (Section \ref{subsec:winner_favored}). The following assumption on the class of transfer rules $\mathcal{T}$ incorporates these two classes as special cases:

\begingroup
\renewcommand\theassumption{T}
\begin{assumption} \label{assum:T}

Let $0 \leq \alpha \leq \beta \leq 1$ and $K \geq 0$ be constants. The class of transfer rules $\mathcal{T}$ is given as follows: $t \in \mathcal{T}$ if and only if for every $i$, $\theta < \bar \theta$, $\theta^w$ and $\theta^l$,
\[
x_i(\theta, \theta^w)=1 \,\text{and}\, x_i(\theta, \theta^l)=0 \implies \alpha t_i(\theta, \theta^w) - \beta t_i(\theta, \theta^l) \leq K.
\]
\end{assumption}
\endgroup

\noindent The limited premium transfers correspond to $\alpha = 0$ and $\beta =1$, and the winner-favored transfers correspond to $\alpha=\beta=1$ and $K=0$.

The seller's problem, studied in subsequent sections, is given as follows.

\medskip

\noindent \textbf{Optimal Transfer Problem \eqref{eq:opt_trans_problem}.} Given an allocation rule $x$, solve
\begin{align} \label{eq:opt_trans_problem}
\sup_t \{ \sum_i \int_\Theta T_i(\theta) dP(\theta): \,\, &(i) \, \text{$(x, t)$ is feasible} \notag \\
&(ii) \, t \in \mathcal{T} \}. \tag{$\mathbf{P}$}
\end{align}

\noindent Our setup extends the standard setup where the seller faces only constraint (i) in problem \eqref{eq:opt_trans_problem}, since constraint (ii) is trivially satisfied when $\alpha=\beta=K=0$.

\section{Main Result} \label{sec:main}

\subsection{Statement} \label{subsec:statement}

This section presents our main result. We first define the class of \textit{win-lose dependent mechanisms}:

\begin{definition} \label{def:win_lose_dependent}

We say that $(x, t)$ is a win-lose dependent mechanism if there exist $t_i^w, t_i^l: \Theta \rightarrow \mathbb{R}$ such that, for every $i$, $\theta$ and $\theta'$,
\begin{align}
t_i(\theta, \theta') &= t_i^w(\theta) x_i(\theta, \theta') + t_i^l(\theta) [1-x_i(\theta, \theta')] \label{eq:win_lose_dependent} \\
\theta-t_i^w(\theta) &\geq -t_i^l(\theta). \label{eq:winning_is_better_than_losing}
\end{align}
If $(x, t)$ is a win-lose dependent mechanism, we say $t$ is a win-lose dependent transfer.
\end{definition}

Equation \eqref{eq:win_lose_dependent} means that each type of bidder's transfers (and hence her payoffs) conditional on winning and losing are independent of the competitor's type report. Specifically, bidder $i$ with type $\theta$ pays $t_i^w(\theta)$ to the seller and receives a payoff of $\theta-t_i^w(\theta)$ if she wins, whereas she pays $t_i^l(\theta)$ and receives $-t_i^l(\theta)$ if she loses. Inequality \eqref{eq:winning_is_better_than_losing} means that the winning payoff is higher than the losing payoff, a common property of most auctions. If inequality \eqref{eq:winning_is_better_than_losing} holds with equality, then $(x, t)$ becomes a \textit{full insurance mechanism} \citep{Bo06} where each bidder's payoff is constant with respect to the competitor's type report. Hence, the win-lose dependent mechanism generalizes the full insurance mechanism by allowing the winning and losing payoffs to differ.

Our main result, Theorem \ref{thm:main}, identifies a class of optimal transfer candidates satisfying two properties: (i) $t$ is win-lose dependent, and (ii) a type $\theta$ bidder's interim worst-case utility equals the sum of the minimum winning probabilities of all types below $\theta$.

\begin{theorem} \label{thm:main}
Suppose $x$ is given, and Assumptions \ref{assum:div}, \ref{assum:X} and \ref{assum:T} hold.
Consider the following problem:

\medskip

\noindent \textnormal{\textbf{Reduced Problem \eqref{eq:reduced_problem}.}} Given an allocation rule $x$, solve
\begin{align} \label{eq:reduced_problem}
\sup_t \Big\{ \sum_i \int_\Theta T_i(\theta) dP(\theta): \,\, &(i) \, \text{$t$ is win-lose dependent} \notag\\
&(ii) \, U_i^{\min}(\theta) = \int_{\underline \theta}^\theta X_i^{\min}(z) dz \quad \text{for every $i$ and $\theta$} \notag \\
&(iii) \, t \in \mathcal{T} \Big\} \tag{$\mathbf{R}$}
\end{align}
where $X_i^{\min}(\theta)$, $U_i^{\min}(\theta)$ and $T_i(\theta)$ are defined in equations \eqref{eq:notation_Xmin}-\eqref{eq:notation_T}.
Then,
\begin{equation} \label{eq:R_in_P}
\text{solution set of problem \eqref{eq:reduced_problem}} \, \subset \, \text{solution set of problem \eqref{eq:opt_trans_problem}}.
\end{equation}
\end{theorem}
\begin{proof}
See Section \ref{subsec:proof}.
\end{proof}

\noindent It can be shown that the inclusion relation \eqref{eq:R_in_P} becomes an equality under suitable regularity conditions on the divergence.\footnote{This result is omitted because of space limitation, but its formal statement and proof are available upon request.} In this case, whenever the reduced problem \eqref{eq:reduced_problem} has a unique solution (which is the case in Sections \ref{subsec:limited_prem}-\ref{subsec:winner_favored}), so does the optimal transfer problem \eqref{eq:opt_trans_problem}.

Theorem \ref{thm:main} essentially reduces an infinite-dimensional problem \eqref{eq:opt_trans_problem} to a two-dimensional problem, which can be solved graphically (Figure \ref{fig:w_l_plane}). Concretely, the ex-ante formulation of the reduced problem \eqref{eq:reduced_problem} is equivalent to the interim formulation of choosing an interim transfer $t_i(\theta, \cdot)$ for each fixed $i$ and $\theta$. We now express the interim formulation in terms of a two-dimensional vector $(w, l)$, where $w$ and $l$ denote the winning and losing payoffs of bidder $i$ with type $\theta$. First, the interim expected revenue can be written as
\begin{align} \label{eq:T_w_l}
T_i(\theta) &= \underbrace{\theta X_i(\theta)}_\text{surplus} - \underbrace{\int_\Theta [ \theta x_i(\theta, \theta') - t_i(\theta, \theta') ] dP(\theta')}_{\substack{\text{bidder's share of the surplus}\\\text{($:=$ bidder's expected payoff under $P$)}}} \notag \\
&= \theta X_i(\theta) - [ w X_i(\theta) + l(1-X_i(\theta)) ].
\end{align}
Equation \eqref{eq:T_w_l} also shows that maximizing the revenue $T_i(\theta)$ is equivalent to minimizing the bidder's share of the surplus $[ w X_i(\theta) + l(1-X_i(\theta)) ]$. Next, since $t$ is win-lose dependent, we have $w \geq l$ (see inequality \eqref{eq:winning_is_better_than_losing}). This implies that the bidder's expected payoff is minimized under a belief which minimizes the winning probability; hence, her interim worst-case utility becomes
\begin{align} \label{eq:Umin_w_l}
U_i^{\min}(\theta) &= \inf_Q \{ w Q\{\theta': x_i(\theta, \theta')=1\} + l Q\{\theta': x_i(\theta, \theta')=0\}: D(Q||P) \leq \eta\} \notag\\
&= w X_i^{\min}(\theta) + l (1-X_i^{\min}(\theta)).
\end{align}
Finally, by Assumption \ref{assum:T}, the condition $t \in \mathcal{T}$ becomes $\alpha (\theta-w)+\beta l \leq K$ because the bidder's winning and losing transfers are $\theta-w$ and $-l$.

\begin{figure} [t]
\centering
\includegraphics[scale=0.9]{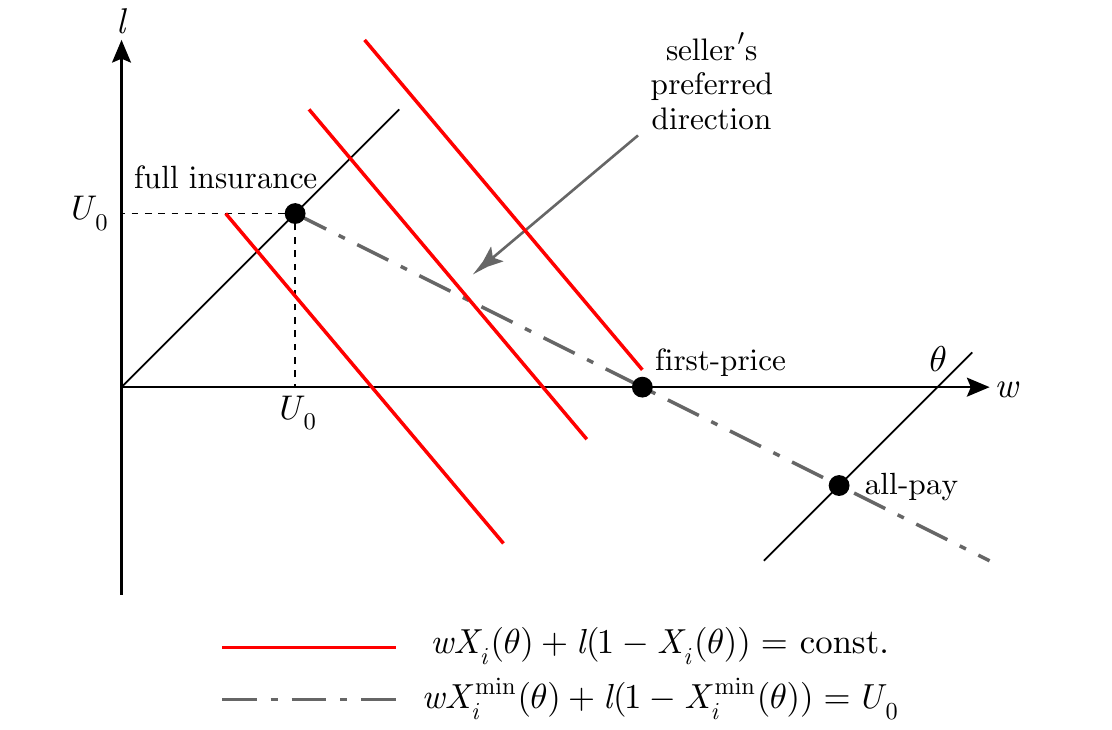}
\caption{\textbf{The interim formulation of the reduced problem.} On the dash-dotted line, the points with $w=l$, $l=0$, and $l=w-\theta$ correspond to the full insurance, first-price, and all-pay auctions, respectively (see Section \ref{sec:constraints_applications} for details).} \label{fig:w_l_plane}
\end{figure}

Thus, the reduced problem \eqref{eq:reduced_problem} simplifies as follows:

\noindent \textbf{Reduced Problem-Interim \eqref{eq:reduced_problem_interim}.} Let $x$ be given. For given $i$ and $\theta$, solve
\begin{align} \label{eq:reduced_problem_interim}
\textrm{inf}_{(w, l)} \{ w X_i(\theta) + l (1-X_i(\theta)) : \,\, &(i) \, w \geq l \notag \\
&(ii) \, w X_i^{\min}(\theta) + l (1-X_i^{\min}(\theta)) = U_0 \notag \\
&(iii) \, \alpha (\theta-w)+\beta l \leq K \},  \tag{$\mathbf{R}$-Int}
\end{align}
where $U_0 := \int_{\underline \theta}^\theta X_i^{\min}(z) dz$ is a constant (since $x$, $i$, and $\theta$ are fixed).  Figure \ref{fig:w_l_plane} illustrates problem \eqref{eq:reduced_problem_interim} on the $(w, l)$-plane.  In Section \ref{sec:constraints_applications}, by solving problem \eqref{eq:reduced_problem_interim}, we derive the optimal transfer rule in two classes $\mathcal{T}$.

\subsection{Proof} \label{subsec:proof}

In this section, we prove Theorem \ref{thm:main}. \cite{Bo06} derive the optimality of full insurance mechanisms by showing that, for a given feasible transfer, the full insurance transfer providing the bidders with the same interim worst-case utility as the given satisfies the feasibility constraint and generates a higher revenue than the original. This argument does not extend to our setup: a win-lose dependent transfer providing the bidders with the same interim worst-case utility as the given is not necessarily feasible. We overcome this difficulty by using an alternative two-step construction. First, we take the conditional expectation of the given transfer with respect to the bidders' winning and losing events; the resulting transfer yields the same revenue as the given for the seller and provides higher interim worst-case utilities than the given for the bidders (Proposition \ref{prop:SOSD}), but needs not be feasible. Second, we adjust the transfer obtained in the first step so that it satisfies the feasibility constraint (Proposition \ref{prop:feasibility}) and generates a higher revenue than the original.

\medskip

We first explain Proposition \ref{prop:feasibility}, which provides a necessary condition and a sufficient condition for a mechanism to be feasible. When bidders are ambiguity averse, the so-called ``payoff equivalence theorem'' \citep{Myer81} no longer holds; i.e., the bidder's interim worst-case utilities under two feasible mechanisms with the same allocation rule need \textit{not} be equal (up to a constant). Thus, the bidder's worst-case utility under a feasible mechanism cannot be expressed in terms of the allocation rule alone, but depends also on the transfer rule.\footnote{\cite{Bodoh12, Bodoh14} shows that under suitable regularity conditions, a type $\theta$ bidder's worst-case utility under a feasible mechanism equals the sum of the winning probabilities under worst-case beliefs (i.e., beliefs minimizing the bidders' expected payoffs) of all types below $\theta$. Since worst-case beliefs depend on the transfer rule, so does the worst-case utility.} Despite these difficulties, the following characterizations hold:

\begin{proposition} [Characterizations of feasible mechanisms] \label{prop:feasibility} \,

\noindent Let $x$ be given. If Assumptions \ref{assum:div} and \ref{assum:X} hold, the following statements hold:

\noindent (i) \citep{Bo06} Suppose that $(x, t)$ is feasible. Then, for every $i$ and $\theta$,
\[
U_i^{\min}(\theta) \geq \int_{\underline \theta}^\theta X_i^{\min}(z) dz.
\]

\noindent (ii) Suppose that $\bar t$ is a win-lose dependent transfer such that for every $i$ and $\theta$,
\begin{equation} \label{eq:feasibility_sufficient}
\bar U_i^{\min}(\theta) = \int_{\underline \theta}^\theta X_i^{\min}(z) dz,
\end{equation}
where $\bar U_i^{\min}(\theta)$ denotes the bidder's interim utility under $(x, \bar t)$. Then, $(x, \bar t)$ is feasible.

\end{proposition}
\begin{proof}
See Appendix \ref{appen:main_proof}.
\end{proof}

Proposition \ref{prop:feasibility} (i) states that a necessary condition for feasibility is that a type $\theta$ bidder's interim worst-case utility exceeds a certain lower bound---the sum of the minimum winning probabilities of all types below $\theta$. This result has been proved by \cite{Bo06}; we provide the proof for completeness. Proposition \ref{prop:feasibility} (ii) states that a sufficient condition for a \textit{win-lose dependent mechanism} to be feasible is that each bidder's interim worst-case utility achieves the lower bound in Proposition \ref{prop:feasibility} (i). Intuitively, under a win-lose dependent mechanism, the bidder's interim worst-case utility is evaluated under the belief minimizing her winning probability (see equation \eqref{eq:Umin_w_l}); hence, the interim worst-case utility equals the sum of \textit{minimum} winning probabilities $X_i^{\min}(z)$ (equation \eqref{eq:feasibility_sufficient}).  Thus, our result partially generalizes the standard sufficiency result \citep{Myer81} in the sense that equation \eqref{eq:feasibility_sufficient} reduces to the standard envelope formula when ambiguity is absent ($\eta=0$).

\medskip

Next, Proposition \ref{prop:SOSD} states that if a payoff schedule $\bar \pi$ is less variable than another payoff schedule $\pi$---in the sense that $\bar \pi$ is the conditional expectation of $\pi$ with respect to coarser information---then $\bar \pi$ yields a higher worst-case utility than $\pi$. Since $\bar \pi$ second-order stochastically dominates $\pi$ in this case,\footnote{To see this, suppose $\bar \pi = \mathbb{E}_P[\pi|\mathcal{E}]$ for some $\sigma$-algebra $\mathcal{E}$. Let $\epsilon := \pi - \bar \pi$. Then,
\[
\mathbb{E}_P[\epsilon|\bar \pi] = \mathbb{E}_P[\mathbb{E}_P[\epsilon|\mathcal{E}] | \bar \pi] = \mathbb{E}_P[ \mathbb{E}_P[\pi|\mathcal{E}]-\bar \pi| \bar \pi ] = 0.
\]
This shows that $\bar \pi$ is a mean-preserving contraction of $\pi$, or equivalently second-order stochastically dominates $\pi$.} Proposition \ref{prop:SOSD} implies that ambiguity averse agents prefer payoff schedules that rank higher in a special kind of the second-order stochastic dominance (SOSD).\footnote{Under slightly stronger assumptions than Assumption \ref{assum:div}, Proposition \ref{prop:SOSD} extends as follows: if $\bar \pi$ second-order stochastically dominates $\pi$, then $\bar \pi$ yields a higher worst-case utility than $\pi$. For our purpose, Proposition \ref{prop:SOSD} suffices.}

\begin{proposition} [Monotonicity with respect to SOSD] \label{prop:SOSD}
Suppose that Assumption \ref{assum:div} holds. Let $\pi: \Theta \rightarrow \mathbb{R}$ be a bounded measurable function. Given a $\sigma$-algebra $\mathcal{E}$, let
\[
\bar \pi:= \mathbb{E}_P[\pi|\mathcal{E}].
\]
Then,
\[
\inf_Q \{ \mathbb{E}_Q[\bar \pi]: D(Q||P) \leq \eta\} \geq \inf_Q \{ \mathbb{E}_Q[\pi]: D(Q||P) \leq \eta\}.
\]
\end{proposition}
\begin{proof}
See Appendix \ref{appen:main_proof}.
\end{proof}

The intuition of Proposition \ref{prop:SOSD} is that the minimum expectation of $\bar \pi$ equals the minimum expectation of $\pi$ over the restricted set of ``coarse'' beliefs---whose likelihood ratios are measurable with respect to the coarser information $\mathcal{E}$---and therefore exceeds the unrestricted minimum. Specifically, given a belief $Q$ in the set of priors, define a coarse belief $\bar Q$ as $d \bar Q := \frac{dQ_{\mathcal{E}}}{dP_{\mathcal{E}}} dP$. Then, properties \textbf{D4}-\textbf{D5} of Assumption \ref{assum:div} imply that the coarse belief $\bar Q$ is closer to $P$ than the original belief $Q$ is, and hence lies in the set of priors: $D(\bar Q||P) \leq D(Q||P) \leq \eta$. Hence
\[
     \mathbb{E}_Q[\bar \pi] = \mathbb{E}_{\bar Q}[\pi]  \geq  \inf_Q \{ \mathbb{E}_Q[\pi]: D(Q||P) \leq \eta\},
\]
where the first equality follows from the basic properties of conditional expectation. Thus, Proposition \ref{prop:SOSD} holds.

\com{The intuition of Proposition \ref{prop:SOSD} is as follows. When a decision maker minimizes the expectation of $\bar \pi$, she behaves as if she minimize the expectation of $\pi$ over a restricted set of ``coarse'' beliefs---whose probability densities relative to $P$ are measurable with respect to the coarser information $\mathcal{E}$---hence overestimating the minimum compared to the unrestricted minimum.}

We now explore the mechanism design implication of Proposition \ref{prop:SOSD}. Corollary \ref{cor:SOSD_implication} states that holding each bidder's interim worst-case utility above a lower bound $\int_{\underline \theta}^\theta X_i^{\min}(z) dz$, the win-lose dependent mechanisms yield greater revenues than others.

\begin{corollary} \label{cor:SOSD_implication}

Given $x$, suppose Assumptions \ref{assum:div}, \ref{assum:X} and \ref{assum:T} hold. Consider $t$ such that
\[
U_i^{\min}(\theta) \geq \int_{\underline \theta}^\theta X_i^{\min}(z) dz \quad \text{and} \quad t \in \mathcal{T}.
\]
Then, there exists a win-lose dependent transfer rule $\hat t$ such that
\[
\widehat T_i(\theta) \geq T_i(\theta), \quad \widehat U_i^{\min}(\theta) = \int_{\underline \theta}^\theta X_i^{\min}(z) dz \quad \text{and} \quad \hat t \in \mathcal{T},
\]
where $\widehat T_i(\theta)$ and $\widehat U_i^{\min}(\theta)$ denote the interim variables under $(x, \hat t)$.

\end{corollary}
\begin{proof}
See Appendix \ref{appen:main_proof}.
\end{proof}

The intuition of Corollary \ref{cor:SOSD_implication} is as follows. Suppose that $\pi$ is the payoff schedule of bidder $i$ with type $\theta$ under $t$ (i.e., $\pi \equiv \theta x_i(\theta, \cdot)-t_i(\theta, \cdot)$). Let $\mathcal{E}$ be the $\sigma$-algebra generated by the bidder's winning and losing events, and $\bar \pi$ be the conditional expectation of $\pi$. Note that $\mathbb{E}_P[\bar \pi]=\mathbb{E}_P[\pi]$ by the law of iterated expectations. Proposition \ref{prop:SOSD} implies that if
\[
\inf_Q \mathbb{E}_Q[\pi] \geq  U_0 := \int_{\underline \theta}^\theta X_i^{\min}(z) dz,
\]
then the following holds:
\[
\begin{array} {l}
\text{There exists $\mathcal{E}$-measurable $\bar \pi$ s.t.}\\
(i) \, \mathbb{E}_P[\bar \pi]=\mathbb{E}_P[\pi] \\
(ii) \, \inf_{Q} \mathbb{E}_Q[\bar \pi] \geq U_0
\end{array} \iff \begin{array} {l}
\text{There exists $\mathcal{E}$-measurable $\hat \pi$ s.t.}\\
(i) \, \mathbb{E}_P[\hat \pi] \leq \mathbb{E}_P[\pi] \\
(ii) \, \inf_{Q} \mathbb{E}_Q[\hat \pi] = U_0,
\end{array}
\]
where the equivalence follows from the usual duality argument. Thus, we can construct a win-lose dependent transfer rule $\hat t$ which (i) extracts a higher revenue than $t$ for the seller, and (ii) guarantees the bidder with the lower bound for the interim worst-case utility:
\begin{align*}
    (i)   &\quad \widehat T_i(\theta) \geq T_i(\theta) \iff  \mathbb{E}_P[\hat \pi] \leq \mathbb{E}_P[\pi] \\
    (ii) &\quad  \widehat U_i^{\min}(\theta) = \int_{\underline \theta}^\theta X_i^{\min}(z) dz \iff  \inf_{Q} \mathbb{E}_Q[ \hat \pi] = U_0 .
\end{align*}

\medskip

\begin{proof} [Proof of Theorem \ref{thm:main}]

First, by Proposition \ref{prop:feasibility} (ii), the constraint set of the reduced problem \eqref{eq:reduced_problem} is contained in the constraint set of the optimal transfer problem \eqref{eq:opt_trans_problem}. Now, let $t$ be a transfer rule in the constraint set of problem \eqref{eq:opt_trans_problem}. By Propostion \ref{prop:feasibility} (i), $t$ satisfies the hypothesis of Corollary \ref{cor:SOSD_implication}. This implies that there exists a transfer rule $\hat t$ in the constraint set of problem \eqref{eq:reduced_problem} which yields a higher revenue than $t$. Thus, Theorem \ref{thm:main} holds.

\end{proof}

\section{Optimal transfer rules: limited premium and winner-favored transfers} \label{sec:constraints_applications}

Theorem \ref{thm:main} simplifies the optimal transfer problem \eqref{eq:opt_trans_problem} to the reduced problem \eqref{eq:reduced_problem}. Here, by solving the reduced problem \eqref{eq:reduced_problem}, we find the optimal transfer rule in two classes: (i) the limited premium transfers (Section \ref{subsec:limited_prem}) and (ii) the winner-favored transfers (Section \ref{subsec:winner_favored}).

\subsection{Limited premium transfers} \label{subsec:limited_prem}

\begin{table}
{\small
\begingroup
\setstretch{1.4}
\begin{tabular}{c|c|c|c||c|c|c}
  \hline
   & FIA & FPA & SPA & APA & WoA & SLA \\
   \hline
  Win & $t^w >  0$ & $t^w >0$ & $t^w >0$ & $t^w >0$ & $t^w >0 $& $t^w =0$ \\
  \hline
  Lose & $t^l < 0$ & $t^l =0$ & $t^l =0$ & $t^l>0$ & $ t^l>0 $ & $ t^l>0$\\
    \hline
   Comparison & $-t^l = \theta-t^w$ & $-t^l =0$  & $-t^l =0$  & $0 < t^w = t^l$ & $0< t^w < t^l$ & $0=t^w < t^l$ \\
  \hline
  Constraint  &  \multicolumn{3}{c||}{$-t^l \leq K$} & \multicolumn{3}{c}{$t^w \leq t^l$}\\
  \hline
\end{tabular}
\endgroup
}
\caption{\textbf{Comparison of transfers.} FIA, FPA, SPA, APA, WoA, and SLA refer to the full insurance auction, first-price auction, second-price auction, all-pay auction, war of attrition and sad loser auction, respectively. Each type of bidder's winning and losing transfers are denoted by $t^w$ and $t^l$, respectively.}
\end{table}

As mentioned in the introduction, in many standard formats such as the first-price and second-price auctions, bidders receive no premiums. Even when premiums are provided, their amounts are mostly limited and hence insufficient to fully insure bidders (\citealp{Mil04, Goe04, Hu11, Brun14, Hu18}; see Section \ref{sec:literature}). This motivates the following class of \textit{limited premium transfers}:

\begin{definition} \label{def:limited_prem}
    Let $K \geq 0$. We say $(x, t)$ is a $K$-limited premium mechanism if the premium provided to the loser is limited by $K$: for every $i$, $\theta < \bar \theta$, $\theta^w$ and $\theta^l$,
    \begin{equation} \label{eq:limited_prem}
    x_i(\theta, \theta^w)=1 \text{ and } x_i(\theta, \theta^l)=0 \implies -t_i(\theta, \theta^l) \leq K.
    \end{equation}
    If $(x, t)$ is a $K$-limited premium mechanism, we say $t$ is a $K$-limited premium transfer. Especially, if $K=0$, we simply say a no premium mechanism (or transfer).
\end{definition}

\noindent As mentioned, this class satisfies Assumption \ref{assum:T} with $\alpha=0$ and $\beta=1$.

We solve the optimal transfer problem \eqref{eq:opt_trans_problem} in the class of $K$-limited premium transfers. The interim formulation of the reduced problem \eqref{eq:reduced_problem_interim} becomes:

\noindent \textbf{Reduced Problem, Limited Premium.} Let $x$ be given. For given $i$ and $\theta$, solve
\begin{align}
\textrm{inf}_{(w, l)} \{ w X_i(\theta) + l (1-X_i(\theta)) : \,\, &(i) \, w \geq l \notag \\
&(ii) \, w X_i^{\min}(\theta) + l (1-X_i^{\min}(\theta)) = U_0 \notag \\
&(iii) \, l \leq K \}, \label{eq:two_dim_limited_prem}
\end{align}
where $U_0 := \int_{\underline \theta}^\theta X_i^{\min}(z) dz$. Constraint (iii) corresponds to condition \eqref{eq:limited_prem}.

\begin{figure} [t]
\centering
\includegraphics[scale=1]{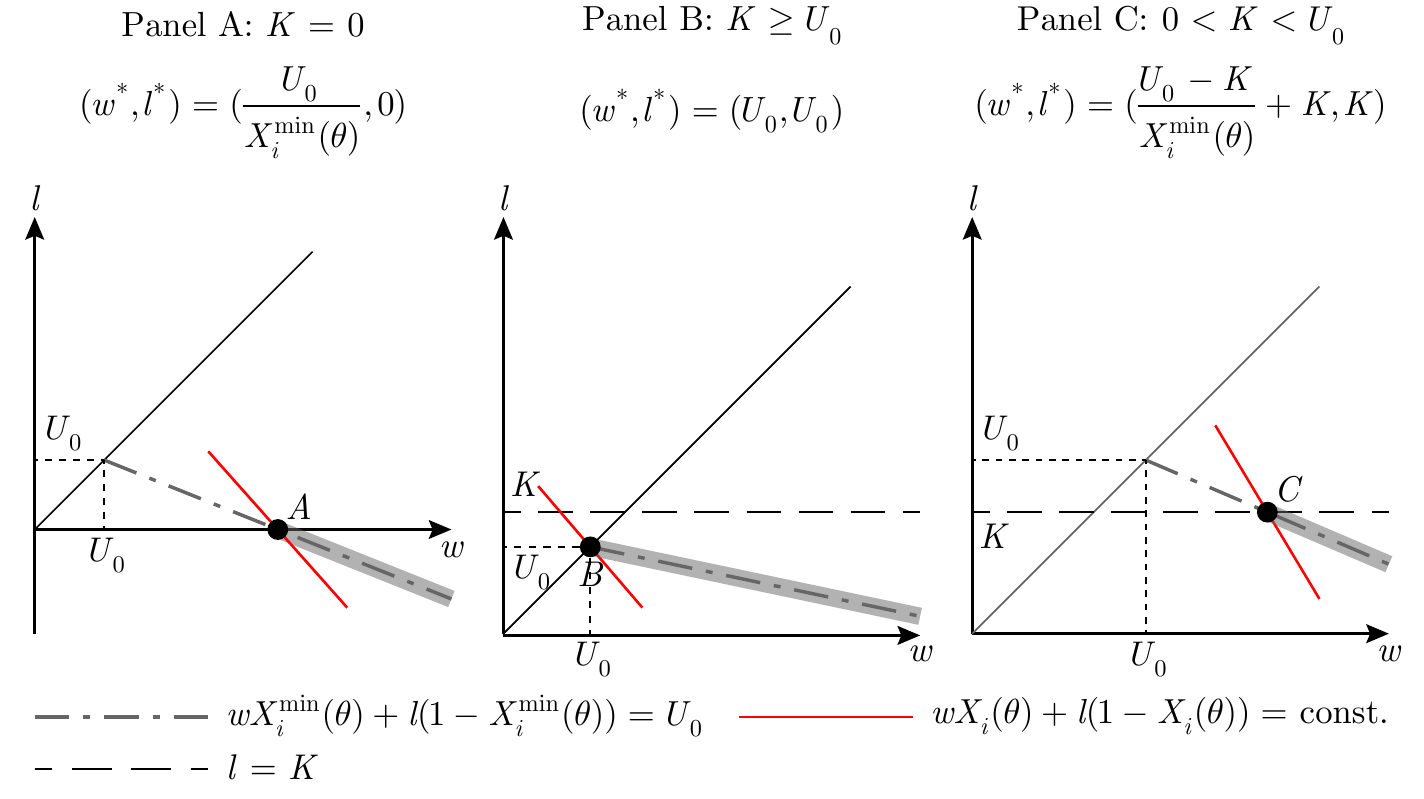}
\caption{\textbf{The interim formulation of the reduced problem.} Panels A, B and C correspond to the following cases: (A) $K=0$, (B) $U_0 \leq K$ and (C) $0<K<U_0$. Points $A$, $B$ and $C$ are the solutions to problem \eqref{eq:two_dim_limited_prem} in cases A, B and C. The shaded line represents the constraint set. Inequality \eqref{eq:w_l_intuition} means that the level curves of the objective function (the solid red lines) are steeper than the constraint line (the gray dash-dotted lines).} \label{fig:solution_w_l_limited_prem}
\end{figure}

Figure \ref{fig:solution_w_l_limited_prem} illustrates problem \eqref{eq:two_dim_limited_prem} and its solution $(w^*, l^*)$. We show that the solution is the point minimizing the difference between $w$ and $l$, i.e., the point closest to the $45^\circ$ line. The ambiguity averse bidder underestimates the winning probability (i.e., $X_i^{\min}(\theta)<X_i(\theta)$) and overestimates the losing probability (i.e., $1-X_i^{\min}(\theta)>1-X_i(\theta)$) relative to the ambiguity neutral seller. Hence, the ratio of the winning probability to the losing probability under the bidder's worst-case belief---the relative effectiveness of $w$ in terms of $l$ in providing the bidder's worst-case utility---is less than the corresponding ratio under the seller's belief---the relative cost of $w$ in terms of $l$ to the seller:
\begin{equation} \label{eq:w_l_intuition}
\begin{array}{c}
\text{relative effectiveness} \\ \text{of $w$ in terms of $l$}
\end{array} = \frac{X_i^{\min}(\theta)}{1-X_i^{\min}(\theta)} < \frac{X_i(\theta)}{1-X_i(\theta)} = \begin{array}{c}
\text{relative cost} \\ \text{of $w$ in terms of $l$}.
\end{array}
\end{equation}
This means that holding the bidder's worst-case utility fixed (constraint (ii)), the cost of substituting $w$ for $l$ exceeds the benefit; hence, the seller tries to reduce $w$ and raise $l$ as much as possible. Thus, the point with the minimum difference between $w$ and $l$ is optimal.

For instance, consider the case $K = 0$ (Panel A). Since the bidder receives no premiums, her losing payoff is at most $0$; hence, $l^*=0$ (point $A$). Next, suppose $K > 0$. For each $i$ and $\theta$ such that $U_0 := \int_{\underline \theta}^\theta X_i^{\min}(z) dz \leq K$ (Panel B), the seller fully insures the bidder, i.e., $w^*=l^*$ (point $B$). In contrast, if $U_0 := \int_{\underline \theta}^\theta X_i^{\min}(z) dz > K$ (Panel C), unable to provide full insurance, the seller provides the maximum available premium for the loser; hence, $l^*=K$ (point $C$). Thus, there exists a threshold type $\theta_i^K$ such that (i) bidder $i$ with type $\theta \leq \theta_i^K$ receives full insurance, whereas (ii) bidder $i$ with type $\theta > \theta_i^K$ receives only \textit{partial} insurance, i.e., her winning payoff is strictly higher than her losing payoff.

As a result, the optimal transfer rule is given as follows:

\begin{proposition} \label{prop:limited_prem}

Suppose Assumption \ref{assum:div} holds. Let $x$ be an allocation rule satisfying Assumptions \ref{assum:X}, and $\mathcal{T}$ be the class of $K$-limited premium transfers. Define $t^*$ as
\[
    t^*_i(\theta, \theta') := t_i^w(\theta) x_i(\theta, \theta') + t_i^l(\theta) [1-x_i(\theta, \theta')],
\]
where $t_i^w, t_i^l: \Theta \rightarrow \mathbb{R}$ are given as follows:\footnote{The fraction in equation \eqref{eq:tw_tl_limited_prem} is well-defined because, by the definition of $\theta_i^K$ and Assumption \ref{assum:X} (iii), we have $X_i^{\min}(\theta)>0$ whenever $\theta>\theta_i^K$.}
\begin{equation} \label{eq:tw_tl_limited_prem}
  (t_i^w(\theta), t_i^l(\theta))  := \begin{cases}
                                     (\theta - \int_{\underline \theta}^\theta X^{\min}_i(z) dz, -\int_{\underline \theta}^\theta X_i^{\min}(z)dz)  & \mbox{if }  \theta \leq \theta^K_{i} \\
                                     (\theta-K-\frac{1}{X_i^{\min}(\theta)}(\int_{\underline \theta}^\theta X_i^{\min}(z)dz - K ), -K)  & \mbox{if } \theta > \theta^K_{i},
                               \end{cases}
\end{equation}
and
\begin{equation} \label{eq:threshold}
\theta^K_i := \sup \{ \theta: \int_{\underline \theta}^{\theta} X_i^{\min}(z) dz \leq K \}.
\end{equation}
Then, $t^*$ is the unique solution to the reduced problem \eqref{eq:reduced_problem}, and hence a solution to the optimal transfer problem \eqref{eq:opt_trans_problem}.

\end{proposition}

Of special importance is the case where $x$ is the efficient allocation rule, which grants the object to the highest-valuation bidder (ties are randomly broken). Since $(x, t^*)$ is symmetric in this case, we drop the subscript $i$. Corollary \ref{cor:FPA_FIA} states that depending on $K$, the first-price auction, a full insurance mechanism, and a hybrid between these formats become optimal (see Figure \ref{fig:tw_tl_graph} and Table \ref{tab:summary_corollary}).

\begin{corollary} \label{cor:FPA_FIA}
In Proposition \ref{prop:limited_prem}, suppose that $x$ is the efficient allocation rule. Let
\[
\bar K := \int_{\underline \theta}^{\bar \theta} X^{\min}(z) dz.
\]
Then, the following statements hold:

\noindent (i) If $K = 0$, the first-price auction implements $(x, t^*)$.

\noindent (ii) If $K \geq \bar K$, $(x, t^*)$ is a full insurance auction.

\noindent (iii) If $0 < K < \bar K$, $(x, t^*)$ is a hybrid of the first-price and the full insurance auctions.

\end{corollary}

\begin{figure} [t]
\centering
\includegraphics[scale=1]{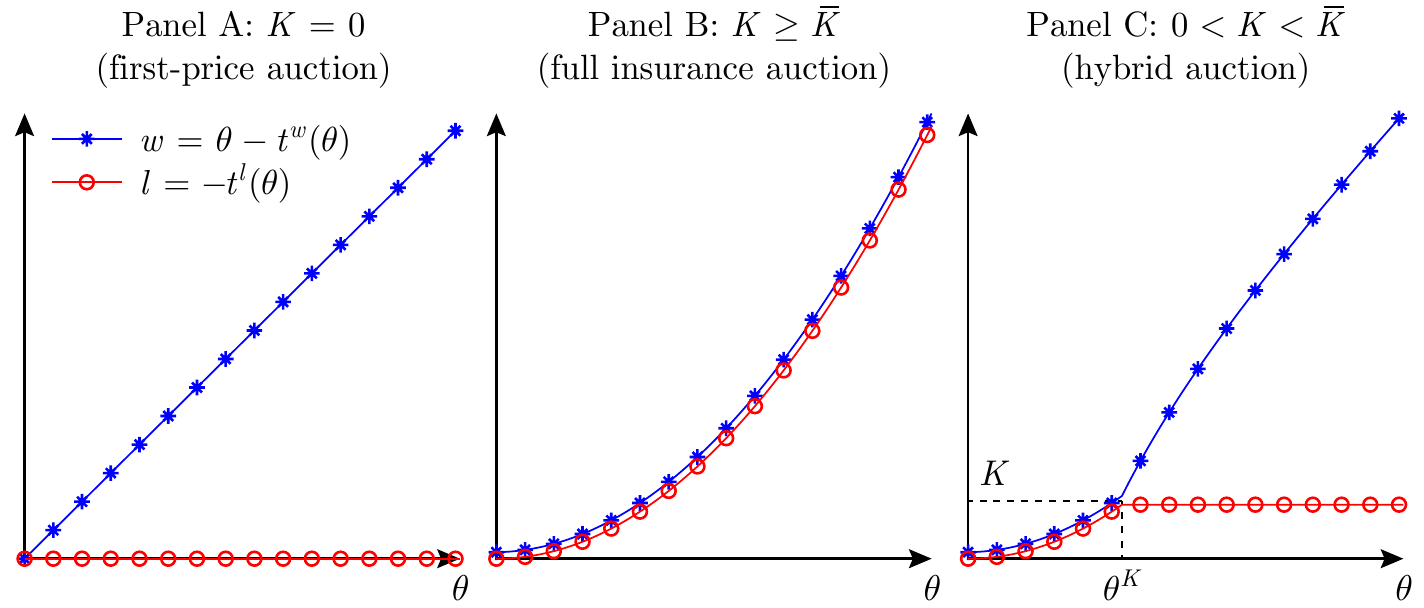}
\caption{\textbf{The bidder's winning and losing payoffs under the optimal transfer in Corollary \ref{cor:FPA_FIA}.} Horizontal axes represent the type $\theta$. Starred lines represent the winning payoff $\theta-t^w(\theta)$, and circled lines represent the losing payoff $-t^l(\theta)$.} \label{fig:tw_tl_graph}
\end{figure}

\begin{table}
  \centering
  {\small \renewcommand{\arraystretch}{1.25}
  \begin{tabular}{l|l|l}
    \hline
     & Optimal transfer & Threshold type \\
     \hline
    $K=0$
    & $t^w(\theta) = \begin{cases}
	\theta &\text{if $\theta \leq \theta^K$}\\
	\theta - \int_{\underline \theta}^\theta \frac{X^{\min}(z)}{X^{\min}(\theta)} dz &\text{if $\theta > \theta^K$}
\end{cases} $
    & $\theta^K = \sup \{ \theta: X^{\min}(\theta) = 0 \}$  \\
    & $t^l(\theta) = 0$ & \\
    \hline
    $K \geq \bar K$ &  $t^w(\theta) = \theta - \int_{\underline \theta}^\theta X^{\min}(z) dz$
 &  $\theta^K = \bar \theta$ \\
    &  $ t^l(\theta) = - \int_{\underline \theta}^\theta X^{\min}(z) dz$   &  \\
    \hline
    $0 < K < \bar K$  & Equation \eqref{eq:tw_tl_limited_prem} &  Equation \eqref{eq:threshold} \\
    \hline
  \end{tabular}
  }
  \caption{\textbf{The optimal transfer in Corollary \ref{cor:FPA_FIA}.} Since $x$ is the efficient allocation rule, we have $X^{\min}(\theta) = \inf_Q \{ Q([\underline \theta, \theta]): D(Q||P) \leq \eta\}$.} \label{tab:summary_corollary}

\end{table}


In the no premium case ($K=0$), the optimal mechanism $(x, t^*)$ has two properties: first, it is win-lose dependent; second, each bidder's losing transfer is zero. A feasible and efficient mechanism with these properties is precisely the first-price auction.\footnote{Alternatively, one can directly verify that the winning transfer $t^w(\theta)$ in Table \ref{tab:summary_corollary} is an equilibrium bidding strategy of the first-price auction \citep{Lo98}.} Corollary \ref{cor:FPA_FIA} (ii) reproduces \cites{Bo06} result as explained earlier. In addition, this result identifies the exact range of the maximum available premium $K$ for which the full insurance mechanism is optimal.

When the seller can provide some premiums but cannot fully insure all types ($0 < K < \bar K$), low types are fully insured ($\theta \leq \theta^K$) whereas high types are only partially insured ($\theta > \theta^K$). For high types, the winning and losing transfers (the second line of equation \eqref{eq:tw_tl_limited_prem}) are identical to the equilibrium transfers of a hybrid indirect mechanism where (i) the bidders first compete in the first-price auction with reserve price $\theta^K$, and (ii) then receive a premium of $K$ regardless of winning or losing. Therefore, the optimal mechanism can be regarded as a hybrid of the first-price auction and a full insurance mechanism.

Corollary \ref{cor:FPA_FIA} extends naturally to the case where the allocation rule excludes bidders with valuations below $r \geq 0$. In this case, if $K=0$, the first-price auction with reserve price $r$ implements the optimal mechanism.

\subsection{Winner-favored transfers} \label{subsec:winner_favored}



Winner-favored mechanisms are commonly observed in the context of contests where contestants' bids are sunk. Typically, the loser pays her own bid, while the winner pays her own bid (the all-pay auction), the second-highest bid (the war of attrition; \citealp{Kris97}), or less than her own bid with some reimbursement from the designer---in the extreme case of the full reimbursement, the winner pays nothing (the sad loser auction; \citealp{Ril81}). In contests, the designer often provides reimbursements (or any ``rewards'' other than the auctioned object) to encourage higher bids (\citealp{Kap02, Coh08, Matros09, Matros12, Min20}; see Section \ref{sec:literature}).

\begin{definition} \label{def:winner_favored}
    We say $(x, t)$ is a winner-favored mechanism if each bidder pays smaller amounts to the seller when she wins than loses: for every $i$, $\theta < \bar \theta$, $\theta^w$ and $\theta^l$,
    \begin{equation}\label{eq:winner_favored}
      x_i(\theta, \theta^w) = 1 \,\text{and}\, x_i(\theta, \theta^l) = 0 \implies t_i(\theta, \theta^w) \leq t_i(\theta, \theta^l).
    \end{equation}
    If $(x, t)$ is a winner-favored mechanism, we say $t$ is a winner-favored transfer.
\end{definition}

\noindent  As mentioned, this class satisfies Assumption \ref{assum:T} with $\alpha=\beta=1$ and $K = 0$.


We solve the optimal transfer problem \eqref{eq:opt_trans_problem} in the class of winner-favored transfers. The interim formulation of the reduced problem \eqref{eq:reduced_problem_interim} becomes:

\noindent \textbf{Reduced Problem, Winner-Favored.} Let $x$ be given. For given $i$ and $\theta$, solve
\begin{align}
\textrm{inf}_{(w, l)} \{ w X_i(\theta) + l (1-X_i(\theta)) : \,\, &(i) \, w \geq l \notag \\
&(ii) \, w X_i^{\min}(\theta) + l (1-X_i^{\min}(\theta)) = U_0 \notag \\
&(iii) \, l \leq w - \theta \}, \label{eq:two_dim_winfav}
\end{align}
where $U_0 := \int_{\underline \theta}^\theta X_i^{\min}(z) dz$. Constraint (iii) corresponds to condition \eqref{eq:winner_favored}.

\begin{figure}
\centering
\includegraphics[scale=1]{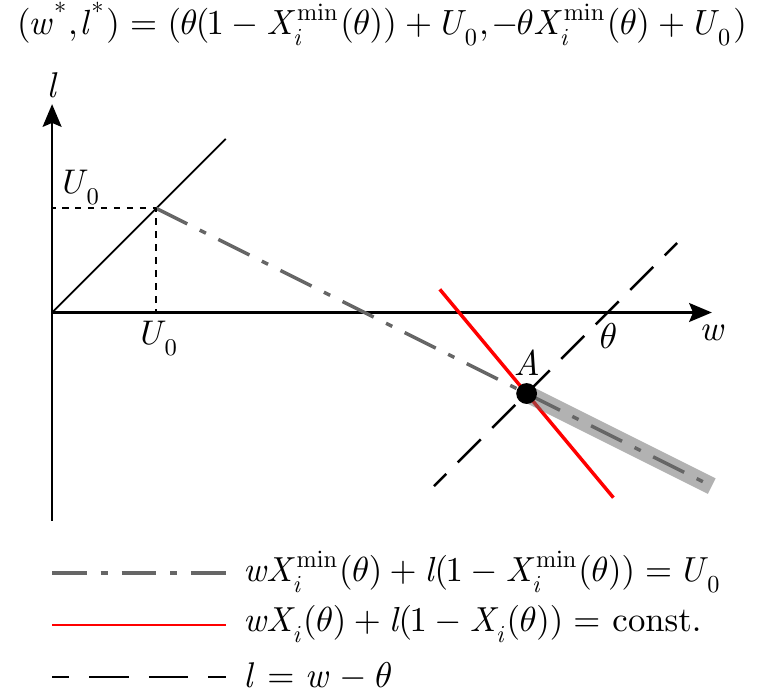}
\caption{\textbf{The interim formulation of the reduced problem.} Point $A$ is the solution $(w^*, l^*)$ to problem \eqref{eq:two_dim_winfav}. The shaded line represents the constraint set.} \label{fig:solution_w_l_APA}
\end{figure}

Figure \ref{fig:solution_w_l_APA} illustrates problem \eqref{eq:two_dim_winfav} and its solution $(w^*, l^*)$ (point $A$). The intuition for the solution is the same as in Section \ref{subsec:limited_prem}: points with smaller differences between $w$ and $l$ yield greater revenues for the seller. Under the constraint that the winning transfer ($=\theta-w$) is less than or equal to the losing transfer ($=-l$), the difference between $w$ and $l$ is minimized when the winning and losing transfers are equal; hence, $l^*=w^*-\theta$.

\begin{proposition} \label{prop:winner_favored}

Suppose Assumption \ref{assum:div} holds. Let $x$ be an allocation rule satisfying Assumptions \ref{assum:X}, and $\mathcal{T}$ be the class of winner-favored transfers. Define $t^*$ as
\[
t_i^*(\theta, \theta') = t_i^w(\theta) x_i(\theta, \theta') + t_i^l(\theta) (1-x_i(\theta, \theta')),
\]
where
\begin{equation} \label{eq:sol_winning_favored}
t_i^w(\theta) = t_i^l(\theta) = \theta X_i^{\min}(\theta) - \int_{\underline \theta}^\theta X_i^{\min}(z) dz.
\end{equation}
Then, $t^*$ is the unique solution to the reduced problem \eqref{eq:reduced_problem}, and hence a solution to the optimal transfer problem \eqref{eq:opt_trans_problem}.

\end{proposition}

Again, consider the case of the efficient allocation rule. The optimal mechanism $(x, t^*)$ satisfies two properties: first, it is win-lose dependent; second, each bidder pays the same amount whether she wins or loses. A feasible and efficient mechanism with these properties is precisely the all-pay auction.\footnote{Alternatively, one can directly verify that expression \eqref{eq:sol_winning_favored} is the equilibrium bidding strategy in the all-pay auction \citep{Stong18}.}

\begin{corollary} \label{cor:APA}
In Proposition \ref{prop:winner_favored}, suppose further that $x$ is the efficient allocation rule. Then, the all-pay auction implements $(x, t^*)$.
\end{corollary}
\noindent As mentioned in the introduction, Corollary \ref{cor:APA} implies that within the class of auctions where the loser pays her own bid, the all-pay auction revenue dominates any auctions (i) in which the winner pays any amounts between the second-highest bid and her own bid, and (ii) in which the winner receives a partial or full reimbursement of her bid.

Corollary \ref{cor:APA} also extends to the case where the allocation rule excludes bidders with valuations below $r \geq 0$. In this case, the optimal mechanism can be implemented by the all-pay auction with reserve price
\[
r \cdot \inf_Q \{ Q([\underline \theta, r]): D(Q||P) \leq \eta \}.
\]

\section{Optimal allocation rules} \label{sec:opt_alloc_rule}

Sections \ref{sec:main}-\ref{sec:constraints_applications} focus on finding the optimal transfer rule, taking the allocation rule as exogenously given. In this section, under a simplifying assumption that the set of priors is defined by the contamination divergence (Example \ref{ex:div}), we endogenously determine the optimal allocation rule. We show that the main results in Section \ref{sec:constraints_applications} remain valid under endogenous allocation: (i) in the class of no premium mechanisms, the first-price auction with a suitable reserve price is optimal; and, (ii) in the class of winner-favored mechanisms, the all-pay auction with a suitable reserve price is optimal (Corollary \ref{cor:optimal_mechanism_endogenous_x}).

In the case of the contamination divergence, the worst-case utility is simply the weighted average of the expected payoff under $P$ and the lowest possible payoff: for a payoff schedule $\pi$,
\[
\inf_Q \{ \mathbb{E}_Q[\pi]: D(Q||P) \leq \eta\} = (1-\eta) \cdot \mathbb{E}_P[\pi] + \eta \cdot \textnormal{ess\,inf}_P \, \pi.
\]
Accordingly, the minimum winning probability simplifies to
\begin{equation} \label{eq:Xmin_contamination}
X_i^{\min}(\theta) = (1-\eta) X_i(\theta) + \eta \cdot \textnormal{ess\,inf}_P \, x_i(\theta, \cdot) = (1-\eta) X_i(\theta) \quad \text{for $\theta<\bar \theta$.}
\end{equation}

Now, we find the optimal mechanism in the class of no premium mechanisms, where the allocation rule is determined endogenously. For a given $x$, by Proposition \ref{prop:limited_prem} and equation \eqref{eq:Xmin_contamination}, the optimal winning and losing transfers are
\[
t_i^w(\theta) = \left\{ \begin{array} {ll}
\theta &\text{if $X_i(\theta)=0$}\\
\theta - \int_{\underline \theta}^\theta \frac{X_i^{\min}(z)}{X_i^{\min}(\theta)} dz = \theta - \int_{\underline \theta}^\theta \frac{X_i(z)}{X_i(\theta)} dz &\text{if $X_i(\theta)>0$}
\end{array} \right. \text{ and } t_i^l(\theta) = 0.
\]
Hence, the optimal ex-ante expected revenue for a given $x$ is
\[
\sum_i \int_\Theta [ t_i^w(\theta) X_i(\theta) + t_i^l(\theta) (1-X_i(\theta)) ] dP(\theta) = \sum_i \int_\Theta [ \theta X_i(\theta) - \int_{\underline \theta}^\theta X_i(z) dz ] dP(\theta).
\]
To find the optimal allocation, we maximize this expression with respect to $x$:

\medskip

\noindent \textnormal{\textbf{Optimal Allocation Problem \eqref{eq:opt_alloc_rule}, No Premium.}}
\begin{equation} \label{eq:opt_alloc_rule}
\sup_x \{ \sum_i \int_\Theta [ \theta X_i(\theta) - \int_{\underline \theta}^\theta X_i(z) dz ] dP(\theta): \text{$x$ satisfies Assumption \ref{assum:X}} \}. \tag{$\mathbf{A}$}
\end{equation}

We solve problem \eqref{eq:opt_alloc_rule} as in \cite{Myer81}. Integrating by parts, we rewrite the objective function as
\begin{align*}
&\sum_i \int_\Theta [ \theta X_i(\theta) - \int_{\underline \theta}^\theta X_i(z) dz ] dP(\theta) = \sum_i \int_\Theta \left( \theta - \frac{1-F(\theta)}{f(\theta)} \right) X_i(\theta) dP(\theta) \\
&= \int_{\Theta^2} \left[ \left(\theta_1 - \frac{1-F(\theta_1)}{f(\theta_1)}\right) x_1(\theta_1, \theta_2) + \left(\theta_2 - \frac{1-F(\theta_2)}{f(\theta_2)}\right) x_2(\theta_2, \theta_1) \right] dP^2(\theta_1, \theta_2),
\end{align*}
where $F$ and $f$ denote the cumulative distribution and probability density functions of $P$. The standard argument yields the following optimal allocation rule:

\begin{proposition} \label{prop:optimal_allocation_rule_contamination}
Let $D$ be the contamination divergence. Suppose that $F$ is absolutely continuous, $f>0$ on $\Theta$ and the function $\theta \mapsto \theta - \frac{1-F(\theta)}{f(\theta)}$ is strictly increasing in $\theta$. Define an allocation rule $x^*$ as
\[
x^*_i(\theta, \theta') := \begin{cases}
	1 &\text{if $\theta > \theta'$ and $\theta \geq r^*$}\\
	\frac{1}{2} &\text{if $\theta = \theta'$ and $\theta \geq r^*$}\\
	0 &\text{otherwise,}
\end{cases}	\quad \text{where $r^* := \sup\{ \theta: \theta - \frac{1-F(\theta)}{f(\theta)} \leq 0\}$.}
\]
Then, $x^*$ is the solution to the optimal allocation problem \eqref{eq:opt_alloc_rule}.
\end{proposition}

The optimal transfer rule corresponding to the optimal allocation rule $x^*$ in Proposition \ref{prop:optimal_allocation_rule_contamination} can be implemented by the first-price auction with reserve price $r^*$ (see Section \ref{subsec:limited_prem}). A similar argument shows that the optimal mechanism in the class of winner-favored mechanisms can be implemented by the all-pay auction with a reserve price. Corollary \ref{cor:optimal_mechanism_endogenous_x} summarizes our results.

\begin{corollary} \label{cor:optimal_mechanism_endogenous_x}

Under the assumptions of Proposition \ref{prop:optimal_allocation_rule_contamination}, the following statements hold:

\noindent (i) The first-price auction with reserve price $r^*$ is optimal in the class of no premium mechanisms satisfying Assumption \ref{assum:X}.

\noindent (ii) The all-pay auction with reserve price $(1-\eta) r^* F(r^*)$ is optimal in the class of winner-favored mechanisms satisfying Assumption \ref{assum:X}.

\end{corollary}

\section{Discussion} \label{sec:diss}

\subsection{Related literature}\label{sec:literature}

\paragraph{Auctions with ambiguity} Our paper primarily contributes to the literature on auctions with ambiguity. Early works study sealed-bid IPV auctions where bidders' preferences follow the MMEU model \citep{Lo98, Bo06, Bodoh12}. Recent papers examine alternative setups: dynamic mechanisms \citep{Bo09, Aus20, Ghosh21}, ambiguity over the joint distribution of valuations and signals \citep{Lao19}, or alternative models of preferences under ambiguity \citep{Chi15, Kocyi20}.\footnote{\cite{Bo09} show that if dynamic mechanisms are allowed, the seller can extract almost all surplus using a variant of the Dutch auction. \cite{Aus20} find that under ambiguity, the strategic equivalence between the first-price auction and the Dutch auction breaks down, because bidders gradually learn about other bidders' types in the latter. \cite{Ghosh21} studies sequential auctions and show that bidders' ambiguity aversion can provide an explanation of the ``declining price anomaly''. \cite{Lao19} study the two-bidder, binary-value and binary-signal setup where the joint distribution of valuations and signals is ambiguous. They find that with common values, the optimal mechanism provides only partial insurance to bidders, in contrast to \cites{Bo06} finding. Some studies assume that bidders' preferences follow \cites{Bew02} model instead of the MMEU model, and show that the second-price auction is desirable in terms of both revenue and efficiency \citep{Chi15, Kocyi20}.} Our paper studies the standard setup of sealed-bid IPV auctions with MMEU bidders which has been analyzed most extensively, but differs from existing works in that we focus on empirically plausible classes of transfers: the limited premium transfers (Section \ref{subsec:limited_prem}) and the winner-favored transfers (Section \ref{subsec:winner_favored}).

\paragraph{Information theory} Recently, concepts in information theory are being used in various economic contexts, such as robust control under model uncertainty in macroeconomics \citep{Han01, Han08}, rational inattention \citep{Sims03, Mat15} and Bayesian persuasion \citep{Gent14, Treu19}. Among them, the most closely related to our paper is the literature on robust control under model uncertainty; this literature studies decision processes maximizing the worst-case performance over a set of priors, given by the relative entropy neighborhood around the ``approximating model''.\footnote{This set of priors is also popular in the literature on robust optimization in operations research  \citep{Ben13, Duchi21} and uncertainty quantification in applied mathematics \citep{Gour20, Dupuis20}.} Our paper generalizes this set of priors by assuming the basic properties of the relative entropy (Assumption \ref{assum:div}) motivated from the statistics and information theory literature \citep{Kull51, Kull59, Ali66}. We use these properties to show that an ambiguity averse agent prefers payoff schedules with less variation (Proposition \ref{prop:SOSD}).

\paragraph{Limited premium mechanisms}

The class of limited premium mechanisms is related to the literature on premium auctions \citep{Mil04, Goe04, Hu11, Brun14, Hu18}. In practice, the premiums are typically limited and insufficient to provide the full insurance. For example, in Amsterdam auctions, bidders first compete in an ascending auction until all but two bidders drop out; then, the two finalists submit bids, and the loser receives a premium equal to 10-30\% of the difference between her bid and the price at which the first stage was terminated \citep{Goe04}. Our approach is in line with this literature in that we allow various amounts of maximum available premiums that the seller can provide.

\paragraph{Winner favored mechanisms}

The class of winner-favored mechanisms is related to the the literature on contests with reimbursements \citep{Coh08, Matros09, Matros12, Min20}, in which contestants' bids (or efforts) are sunk and the seller reimburses contestants' bids to encourage higher bids. In many cases, only the winner receives the reimbursement, or more generally, the winner receives greater amounts of reimbursements than the loser. For instance, the Bush administration's R\&D tax policy provides the winner with tax credits---a form of reimbursement \citep{Matros12, Min20}. As an another example, in Canadian elections, candidates may be reimbursed for up to 25\% of their expenses if the candidate receives 5\% of valid votes in the electoral division. In these cases, each contestant's bid net of the reimbursement---which corresponds to the transfer in our setup---is smaller when she wins than loses.



\subsection{Conclusion} \label{sec:conclusion}

This paper studies the optimal auction design problem where bidders' preferences follow the MMEU model. Assuming that the set of priors consists of beliefs that are close---in terms of a divergence---to the reference belief, we identify a class of optimal transfer candidates satisfying two properties: first, the transfer rule is win-lose dependent; second, a type $\theta$ bidder's interim worst-case utility equals the sum of minimum winning probabilities of all types below $\theta$. Using this result, we find the optimal transfer rule in two empirically plausible classes of transfers: (i) in the class of efficient no premium mechanisms, the first-price auction is optimal; and, (ii) in the class of efficient winner-favored mechanisms, the all-pay auction is optimal.

Our paper focuses on the ambiguity neutral seller. The difficulty in studying the ambiguity averse seller lies in extending Corollary \ref{cor:SOSD_implication}. Our proof of Corollary \ref{cor:SOSD_implication} relies on the fact that given a transfer, its conditional expectation with respect to winning and losing events provides the bidders with higher interim worst-case utilities than the original, while leaving the seller's revenue unchanged. When the seller is ambiguity averse, it is unclear whether the new transfer yields a higher or lower worst-case revenue for the seller than the original. We leave a closer investigation of this case to future research.

\bigskip

\appendix

\noindent\textbf{\Large{}Appendix}{\Large \par} \medskip

\section{Proofs for Section \ref{sec:model}} \label{appen:model}

We use the following well-known fact: if $\mathcal{E} \subset \mathcal{A}$ are $\sigma$-algebras and $P, Q \in \Delta(\Theta, \mathcal{A})$ satisfy $Q \ll P$, then
\begin{equation} \label{eq:restricted_measures_RN}
\frac{dQ_\mathcal{E}}{dP_\mathcal{E}} = \mathbb{E}_P \left[ \left. \frac{dQ}{dP} \right| \mathcal{E} \right].
\end{equation}

\begin{lemma} \label{lem:phi_div}
The $\phi$-divergence (Example \ref{ex:div}) satisfies Assumption \ref{assum:div}.
\end{lemma}
\begin{proof}

As properties \textbf{D1}-\textbf{D3} and \textbf{D5} are clear, we verify \textbf{D4}. Let $\mathcal{E} \subset \mathcal{A}$ be $\sigma$-algebras and $P, Q \in \Delta(\Theta, \mathcal{A})$. By equation \eqref{eq:restricted_measures_RN},
\begin{align*}
D(Q_\mathcal{E}||P_\mathcal{E}) &= \int_\Theta \phi \left( \mathbb{E}_P \left[ \left. \frac{dQ}{dP} \right| \mathcal{E} \right] \right) dP \\
&\leq \int_\Theta \mathbb{E}_P \left[ \left. \phi\left( \frac{dQ}{dP} \right) \right| \mathcal{E} \right] dP = \int_\Theta \phi\left( \frac{dQ}{dP} \right) dP = D(Q||P),
\end{align*}
where the second inequality holds by Jensen's inequality for conditional expectations.

\end{proof}

\begin{lemma} \label{lem:contamination_div}
The contamination divergence (Example \ref{ex:div}) satisfies Assumption \ref{assum:div}. In addition, it generates the contamination model: for $P \in \Delta(\Theta, \mathcal{B})$,
\begin{equation} \label{eq:contamination_div_contamination_model}
\{ Q: D(Q||P) \leq \eta \} = \{ Q: Q = (1-\eta) P + \eta R \text{ for some $R \in \Delta(\Theta, \mathcal{B})$, $R \ll P$} \}.
\end{equation}
\end{lemma}
\begin{proof}

We first verify Assumption \ref{assum:div}. As properties \textbf{D1}-\textbf{D3} and \textbf{D5} are clear, we show \textbf{D4}. Let $\mathcal{E} \subset \mathcal{A}$ be $\sigma$-algebras and $P, Q \in \Delta(\Theta, \mathcal{A})$. By equation \eqref{eq:restricted_measures_RN},
\[
D(Q_\mathcal{E}||P_\mathcal{E}) = \textnormal{ess\,sup}_{P} \left( 1 - \mathbb{E}_P \left[ \left. \frac{dQ}{dP} \right| \mathcal{E} \right] \right) \leq \textnormal{ess\,sup}_{P} \left( 1 - \frac{dQ}{dP} \right) = D(Q||P),
\]
where the second inequality holds because the conditional expectation of a function has a smaller maximum than the original.

Next, we prove equation \eqref{eq:contamination_div_contamination_model}. We show only that the left-hand side is contained in the right-hand side, as the reverse direction is straightforward. Suppose $Q$ satisfies $D(Q||P) \leq \eta$. Define $R := (1/\eta) [Q - (1-\eta)P]$. Rearranging yields $Q = (1-\eta)P + \eta R$. Also, $R \in \Delta(\Theta, \mathcal{B})$ because (i) $R(\Theta) = (1/\eta) [1-(1-\eta)]=1$, and (ii) $D(Q||P) \leq \eta$ implies that $\frac{dR}{dP} = (1/\eta) [\frac{dQ}{dP}-(1-\eta)] \geq 0$. Hence, $Q$ belongs to the right-hand side of equation \eqref{eq:contamination_div_contamination_model}.

\end{proof}

\section{Proofs for Section \ref{subsec:proof}} \label{appen:main_proof}

\begin{proof} [Proof of Proposition \ref{prop:feasibility}]

(i) By incentive compatibility,
\begin{align} \label{eq:necessity_key}
U_i^{\min}(\theta) &\geq \inf_{Q \in \Delta} \int_\Theta [\theta x_i(\hat \theta, \theta') - t_i(\hat \theta, \theta')] dQ(\theta') \notag \\
&= \inf_{Q \in \Delta} \left[  \int_\Theta [ \hat \theta x_i(\hat \theta, \theta') - t_i(\hat \theta, \theta') ] dQ(\theta') + (\theta - \hat \theta) \int_\Theta x_i(\hat \theta, \theta') dQ(\theta')\right] \notag \\
&\geq U_i^{\min}(\hat \theta) +\inf_{Q \in \Delta} (\theta - \hat \theta) \int_\Theta x_i(\hat \theta, \theta') dQ(\theta'),
\end{align}
where the last inequality holds because the minimum of the sum of two functions is greater than the sum of the minimums of two functions.

Now, suppose $\theta > \hat \theta$. Then, inequality \eqref{eq:necessity_key} becomes
\begin{equation} \label{eq:necessity_ineq1}
U_i^{\min}(\theta) \geq U_i^{\min}(\hat \theta) + (\theta - \hat \theta) X_i^{\min}(\hat \theta).
\end{equation}
Interchanging the roles of $\theta$ and $\hat \theta$ in inequality \eqref{eq:necessity_key} yields
\begin{equation} \label{eq:necessity_ineq2}
U_i^{\min}(\hat \theta) \geq U_i^{\min}(\theta) + \inf_{Q \in \Delta} (\hat \theta - \theta) \int_\Theta x_i(\theta, \theta') dQ(\theta') = U_i^{\min}(\theta) + (\hat \theta - \theta) X_i^{\max}(\theta),
\end{equation}
where the second equality holds because $\theta>\hat \theta$. By inequalities \eqref{eq:necessity_ineq1}-\eqref{eq:necessity_ineq2},
\[
X_i^{\min}(\hat \theta) \leq \frac{U_i^{\min}(\theta) - U_i^{\min}(\hat \theta)}{\theta - \hat \theta} \leq X_i^{\max}(\theta).
\]
Thus, $U_i^{\min}$ is Lipschitz and hence absolutely continuous. At a.e. $\theta$, $U_i^{\min}$ is differentiable and $X_i^{\min}$ is continuous (by Assumption \ref{assum:X} (iii)). Taking $\hat \theta \rightarrow \theta$,
\[
X_i^{\min}(\theta) \leq (U_i^{\min})'(\theta) \implies U_i^{\min}(\theta) \geq \int_{\underline \theta}^\theta X_i^{\min}(z) dz.
\]

\noindent (ii) As individual rationality is clear, we show only incentive compatibility. Let $Q^*_{i, \theta} := \argmin_{Q \in \Delta} \int_\Theta x_i(\theta, \theta') dQ(\theta')$ (Lemma \ref{lem:belief_minimizing_winning_prob} shows its existence). Then,
\begin{equation} \label{eq:Xmin_Q*}
X_i^{\min}(\theta) = \int_\Theta x_i(\theta, \theta') dQ_{i, \theta}^*(\theta').
\end{equation}
Since $(x, t)$ is win-lose dependent, each type of bidder's worst-case utility is evaluated under the belief minimizing the winning probability (equation \eqref{eq:Umin_w_l}):
\begin{equation} \label{eq:Umin_Q*}
U_i^{\min}(\theta) = \int_\Theta [ \theta x_i(\theta, \theta') - t_i(\theta, \theta') ] dQ_{i, \theta}^*(\theta').
\end{equation}
We establish incentive compatibility as follows:
\begin{align*}
&U_i^{\min}(\theta) = U_i^{\min}(\hat \theta) + \int_{\hat \theta}^\theta X_i^{\min}(z) dz \geq U_i^{\min}(\hat \theta) + (\theta-\hat \theta) X_i^{\min}(\hat \theta)\\
&= \int_\Theta [\hat \theta x_i(\hat \theta, \theta') - t_i(\hat \theta, \theta')] dQ_{i, \hat \theta}^*(\theta') + (\theta - \hat \theta) \int_\Theta x_i(\hat \theta, \theta') dQ_{i, \hat \theta}^*(\theta')\\
&= \int_\Theta [\theta x_i(\hat \theta, \theta') - t_i(\hat \theta, \theta')] dQ_{i, \hat \theta}^*(\theta') \geq \inf_{Q \in \Delta} \int_\Theta [ \theta x_i(\hat \theta, \theta')-t_i(\hat \theta, \theta') ] dQ(\theta'),
\end{align*}
where the first equality follows by condition \eqref{eq:feasibility_sufficient}, the second inequality by Assumption \ref{assum:X} (iii) (this holds whether $\hat \theta \leq \theta$ or $\hat \theta \geq \theta$), and the third equality by equations \eqref{eq:Xmin_Q*}-\eqref{eq:Umin_Q*}.

\end{proof}

\begin{proof} [Proof of Proposition \ref{prop:SOSD}]

We first prove the following: if $g, h: \Theta \rightarrow \mathbb{R}$ are integrable functions such that $gh$ is integrable, $\bar g := \mathbb{E}_P[g|\mathcal{E}]$ and $\bar h := \mathbb{E}_P[h|\mathcal{E}]$, then
\begin{equation} \label{eq:gh}
\mathbb{E}_P[\bar g h] = \mathbb{E}_P[g \bar h].
\end{equation}
To see this, note that by the law of iterated expectations,
\[
\mathbb{E}_P[\bar g h] = \mathbb{E}_P[ \mathbb{E}_P[\bar g h|\mathcal{E}] ] = \mathbb{E}_P[\bar g \cdot \mathbb{E}_P[h|\mathcal{E}]] = \mathbb{E}_P[\bar g \bar h].
\]
A similar argument yields $\mathbb{E}_P[g \bar h]=\mathbb{E}_P[\bar g \bar h]$. Thus, equation \eqref{eq:gh} holds.

To show Proposition \ref{prop:SOSD}, it suffices to prove the following equation:
\begin{equation} \label{eq:SOSD_proof_key}
\inf_Q \{ \mathbb{E}_Q[\bar \pi]: D(Q||P) \leq \eta \} = \inf_Q \{ \mathbb{E}_Q[\pi]: \text{$D(Q||P) \leq \eta$, \, $\frac{dQ}{dP}$ is $\mathcal{E}$-measurable} \}.
\end{equation}

We first show that the left-hand side exceeds the right-hand side in equation \eqref{eq:SOSD_proof_key}. Let $Q \in \Delta(\Theta, \mathcal{B})$ with $D(Q||P) \leq \eta$ be given. Define $\bar Q \in \Delta(\Theta, \mathcal{B})$ as
\[
d \bar Q := \frac{dQ_{\mathcal{E}}}{dP_{\mathcal{E}}} dP = \mathbb{E}_P[\frac{dQ}{dP}|\mathcal{E}] dP,
\]
where the second equality holds by equation \eqref{eq:restricted_measures_RN}. By construction, the restrictions of $\bar Q$ and $Q$ to $\mathcal{E}$ coincide, i.e., $\bar Q_\mathcal{E}=Q_\mathcal{E}$. By \textbf{D4}-\textbf{D5} of Assumption \ref{assum:div},
\[
D(\bar Q||P) = D(\bar Q_{\mathcal{E}}||P_\mathcal{E}) = D(Q_{\mathcal{E}}||P_{\mathcal{E}}) \leq D(Q||P) \leq \eta.
\]
Thus, by equation \eqref{eq:gh}, we obtain the desired inequality:
\[
\mathbb{E}_Q[\bar \pi] = \mathbb{E}_P[\bar \pi \frac{dQ}{dP}] = \mathbb{E}_P[\pi \frac{d \bar Q}{dP}] = \mathbb{E}_{\bar Q}[\pi] \geq \text{right-hand side of equation \eqref{eq:SOSD_proof_key}}.
\]

It remains to prove the reverse inequality. Consider $Q \in \Delta(\Theta, \mathcal{B})$ such that $D(Q||P) \leq \eta$ and $\frac{dQ}{dP}$ is $\mathcal{E}$-measurable. Since $\mathbb{E}_P[\frac{dQ}{dP}|\mathcal{E}] = \frac{dQ}{dP}$, by equation \eqref{eq:gh},
\[
\mathbb{E}_Q[\pi] = \mathbb{E}_P[ \pi \frac{dQ}{dP} ] = \mathbb{E}_P[\bar \pi \frac{dQ}{dP}] = \mathbb{E}_Q[\bar \pi] \geq \text{left-hand side of equation \eqref{eq:SOSD_proof_key}}.
\]

\end{proof}

\begin{proof} [Proof of Corollary \ref{cor:SOSD_implication}]

To find $\hat t$, for each $i$ and $\theta$, we construct an interim transfer
\[
\hat t_i(\theta, \cdot) = \hat t_i^w(\theta) x_i(\theta, \cdot) + \hat t_i^l(\theta) (1-x_i(\theta, \cdot))
\]
satisfying the following properties:
\begin{align}
&(i) \, \widehat T_i(\theta) \geq T_i(\theta) \label{eq:main_hat_T} \\
&(ii) \, \widehat U_i^{\min}(\theta) = \int_{\underline \theta}^\theta X_i^{\min}(z) dz \label{eq:main_hat_U} \\
&(iii) \, \alpha \hat t_i^w(\theta)-\beta \hat t_i^l(\theta) \leq K \quad \text{if $\theta<\bar \theta$} \label{eq:main_hat_constraint} \\
&(iv) \, \theta - \hat t_i^w(\theta) \geq -\hat t_i^l(\theta). \label{eq:main_hat_w>=l}
\end{align}
We divide into three cases: (A) $0<X_i(\theta)<1$, (B) $X_i(\theta)=0$ and (C) $X_i(\theta)=1$. Case A presents our main idea, while Cases B and C deal with boundary cases.

\noindent \textbf{Case A:} \textit{if $0<X_i(\theta)<1$.} We proceed in two steps ($t \rightarrow \bar t \rightarrow \hat t$).

\noindent \textbf{Case A-Step 1.} We find $\bar t_i(\theta, \cdot) := \bar t_i^w(\theta) x_i(\theta, \cdot) + \bar t_i^l(\theta) (1-x_i(\theta, \cdot))$ satisfying
\begin{align}
&(i) \, \bar T_i(\theta) \geq T_i(\theta) \label{eq:main_bar_T} \\
&(ii) \, \bar U_i^{\min}(\theta) \geq \int_{\underline \theta}^\theta X_i^{\min}(z) dz \label{eq:main_bar_U}\\
&(iii) \, \alpha \bar t_i^w(\theta) - \beta \bar t_i^l(\theta) \leq K \label{eq:main_bar_constraint}\\
&(iv) \, \theta - \bar t_i^w(\theta) \geq -\bar t_i^l(\theta), \label{eq:main_bar_w>=l}
\end{align}
where $\bar U_i^{\min}(\theta)$ and $\bar T_i(\theta)$ denote the interim variables under $(x, \bar t)$.

Define $\bar t_i^w(\theta)$ and $\bar t_i^l(\theta)$ as
\[
\bar t_i^w(\theta) = \int_{\{\theta^w: x_i(\theta, \theta^w)=1\}} t_i(\theta, \theta^w) \frac{dP(\theta^w)}{X_i(\theta)}, \, \bar t_i^l(\theta) =  \int_{\{\theta^l: x_i(\theta, \theta^l)=0\}} t_i(\theta, \theta^l) \frac{dP(\theta^l)}{1-X_i(\theta)}.
\]
Let $\mathcal{E}_{i, \theta}$ be the $\sigma$-algebra generated by the winning and losing events:
\[
\mathcal{E}_{i, \theta} = \sigma( \{ \theta^w: x_i(\theta, \theta^w) = 1 \}, \{ \theta^l: x_i(\theta, \theta^l) = 0 \} ).
\]
Then, $\bar t_i(\theta, \cdot)$ is the conditional expectation of $t_i(\theta, \cdot)$ with respect to $\mathcal{E}_{i, \theta}$:
\begin{equation} \label{eq:cond_expectation}
\bar t_i(\theta, \theta') = \mathbb{E}_P[t_i(\theta, \cdot)|\mathcal{E}_{i, \theta}](\theta') \quad \text{a.e. $\theta'$}.
\end{equation}

Now, we verify that $\bar t_i(\theta, \cdot)$ satisfies properties \eqref{eq:main_bar_T}-\eqref{eq:main_bar_w>=l}. The first property \eqref{eq:main_bar_T} follows directly from equation \eqref{eq:cond_expectation}:
\begin{align*}
\bar T_i(\theta) &:= \int_\Theta \bar t_i(\theta, \theta') dP(\theta') =\int_\Theta \mathbb{E}_P[t_i(\theta, \cdot)|\mathcal{E}_{i, \theta}] dP(\theta') = \mathbb{E}_P[t_i(\theta, \cdot)] = T_i(\theta).
\end{align*}

The second property \eqref{eq:main_bar_U} follows by equation \eqref{eq:cond_expectation} and Proposition \ref{prop:SOSD}:
\begin{align*}
&\bar U_i^{\min}(\theta) := \inf_Q \{ \int_\Theta [\theta x_i(\theta, \theta') - \bar t_i(\theta, \theta')] dQ(\theta'): D(Q||P) \leq \eta\}\\
&= \inf_Q \{ \int_\Theta \mathbb{E}_P[\theta x_i(\theta, \cdot) - t_i(\theta, \cdot)|\mathcal{E}_{i, \theta}](\theta') dQ(\theta'): D(Q||P) \leq \eta\}\\
&\geq U_i^{\min}(\theta) \geq \int_{\underline \theta}^\theta X_i^{\min}(z) dz.
\end{align*}

We verify the third property \eqref{eq:main_bar_constraint} as follows. By Assumption \ref{assum:T}, for $\theta^w$ and $\theta^l$ with $x_i(\theta, \theta^w)=1$ and $x_i(\theta, \theta^l)=0$,
\[
\alpha t_i(\theta, \theta^w) - \beta t_i(\theta, \theta^l) \leq K \implies \alpha \sup_{\theta^w} t_i(\theta, \theta^w) - \beta \inf_{\theta^l} t_i(\theta, \theta^l) \leq K.
\]
Since $\bar t_i^w(\theta) \leq \sup_{\theta^w} t_i(\theta, \theta^w)$ and $\bar t_i^l(\theta) \geq \inf_{\theta^l} t_i(\theta, \theta^l)$, property \eqref{eq:main_bar_constraint} follows.

Finally, if the fourth property \eqref{eq:main_bar_w>=l} holds, the construction is complete. Otherwise, we adjust $\bar t_i(\theta, \cdot)$ using Lemma \ref{lem:w>=l} to obtain $\widetilde t_i(\theta, \cdot)$ which satisfies properties \eqref{eq:main_bar_T}-\eqref{eq:main_bar_w>=l}, and then proceed to Step 2 with $\widetilde t_i(\theta, \cdot)$ instead of $\bar t_i(\theta, \cdot)$.

\noindent \textbf{Case A-Step 2.} Define $\hat t_i^w(\theta)$ and $\hat t_i^l(\theta)$ as
\begin{align*}
\hat t_i^w(\theta) &= \bar t_i^w(\theta) + [ \bar U_i^{\min}(\theta) - \int_{\underline \theta}^\theta X_i^{\min}(z) dz ]\\
\hat t_i^l(\theta) &= \bar t_i^l(\theta) + [ \bar U_i^{\min}(\theta) - \int_{\underline \theta}^\theta X_i^{\min}(z) dz ].
\end{align*}
We verify that $\hat t_i(\theta, \cdot)$ satisfies the desired properties \eqref{eq:main_hat_T}-\eqref{eq:main_hat_w>=l}.

The first property \eqref{eq:main_hat_T} holds because properties \eqref{eq:main_bar_T} and \eqref{eq:main_bar_U} imply that
\[
\widehat T_i(\theta) = \bar T_i(\theta) + [ \bar U_i^{\min}(\theta) - \int_{\underline \theta}^\theta X_i^{\min}(z) dz ] \geq T_i(\theta).
\]

To prove the second property \eqref{eq:main_hat_U}, note that
\[
\theta x_i(\theta, \theta')-\hat t_i(\theta, \theta') = [ \theta x_i(\theta, \theta')-\bar t_i(\theta, \theta') ] - [ \bar U_i^{\min}(\theta) - \int_{\underline \theta}^\theta X_i^{\min}(z) dz ].
\]
It follows that
\[
\widehat U_i^{\min}(\theta) = \bar U_i^{\min}(\theta) - [ \bar U_i^{\min}(\theta) - \int_{\underline \theta}^\theta X_i^{\min}(z) dz ] = \int_{\underline \theta}^\theta X_i^{\min}(z) dz.
\]

We verify the third property \eqref{eq:main_hat_constraint} as follows:
\begin{align*}
\alpha \hat t_i^w(\theta) - \beta \hat t_i^l(\theta) = \alpha \bar t_i^w(\theta) - \beta \bar t_i^l(\theta) - (\beta - \alpha) [ \bar U_i^{\min}(\theta) - \int_{\underline \theta}^\theta X_i^{\min}(z) dz ] \leq K,
\end{align*}
where the last inequality holds because $\alpha \leq \beta$ (see Assumption \ref{assum:T}).

The final property \eqref{eq:main_hat_w>=l} follows immediately from property \eqref{eq:main_bar_w>=l}.

\medskip

\noindent \textbf{Case B:} \textit{if $X_i(\theta)=0$.} Let $\hat t_i^w(\theta) = \hat t_i^l(\theta) := 0$. First, property \eqref{eq:main_hat_T} holds because
\[
\widehat T_i(\theta) = 0 = \theta X_i(\theta) = \mathbb{E}_P[\theta x_i(\theta, \cdot)-t_i(\theta, \cdot)]+T_i(\theta) \geq U_i^{\min}(\theta) + T_i(\theta) \geq T_i(\theta).
\]

Next, we verify property \eqref{eq:main_hat_U}. Since $X_i(\theta)=0$ implies $X_i^{\min}(\theta)=0$,
\[
\widehat U_i^{\min}(\theta) = (\theta-\hat t_i^w(\theta)) \cdot X_i^{\min}(\theta) + (-\hat t_i^l(\theta)) \cdot (1-X_i^{\min}(\theta)) = \theta X_i^{\min}(\theta) = 0.
\]
By Assumption \ref{assum:X} (iii), $X_i^{\min}(z) = 0$ for $z \leq \theta$. Hence, $\widehat U_i^{\min}(\theta) = 0 = \int_{\underline \theta}^\theta X_i^{\min}(z) dz$.

The remaining properties \eqref{eq:main_hat_constraint} and \eqref{eq:main_hat_w>=l} hold by construction.

\medskip

\noindent \textbf{Case C:} \textit{if $X_i(\theta)=1$.} Let $\hat t_i^w(\theta) := \theta - \int_{\underline \theta}^\theta X_i^{\min}(z) dz$ and $\hat t_i^l(\theta) := -\int_{\underline \theta}^\theta X_i^{\min}(z) dz$. Then
\begin{equation} \label{eq:main_X=1_full_insurance}
\theta x_i(\theta, \theta') - \hat t_i(\theta, \theta') = \int_{\underline \theta}^\theta X_i^{\min}(z) dz \quad \text{for every $\theta'$.}
\end{equation}

To verify the first property \eqref{eq:main_hat_T}, observe the following:
\begin{align}
&\theta X_i(\theta)-\widehat T_i(\theta) = \int_{\underline \theta}^\theta X_i^{\min}(z) dz \label{eq:main_X=1_T1} \\
&\theta X_i(\theta) \geq U_i^{\min}(\theta) + T_i(\theta) \geq \int_{\underline \theta}^\theta X_i^{\min}(z) dz + T_i(\theta), \label{eq:main_X=1_T2}
\end{align}
where equation \eqref{eq:main_X=1_T1} follows from equation \eqref{eq:main_X=1_full_insurance}. Comparing equation \eqref{eq:main_X=1_T1} with inequality \eqref{eq:main_X=1_T2} yields $\widehat T_i(\theta) \geq T_i(\theta)$.

The second property \eqref{eq:main_hat_U} follows by equation \eqref{eq:main_X=1_full_insurance}. The third property \eqref{eq:main_hat_constraint} holds automatically, since $\theta = \bar \theta$ by Assumption \ref{assum:X} (ii). The final property \eqref{eq:main_hat_w>=l} holds by construction.

\end{proof}

\begin{lemma} \label{lem:w>=l}

Suppose that the assumptions of Corollary \ref{cor:SOSD_implication} hold. Fix $i$ and $\theta$ such that $0<X_i(\theta)<1$. Consider an interim transfer
\[
\bar t_i(\theta, \cdot) := \bar t_i^w(\theta) x_i(\theta, \cdot) + \bar t_i^l(\theta) (1-x_i(\theta, \cdot))
\]
satisfying $\alpha \bar t_i^w(\theta) - \beta \bar t_i^l(\theta) \leq K$ and $\theta - \bar t_i^w(\theta) < -\bar t_i^l(\theta)$. Then, there exists
\[
\widetilde t_i(\theta, \cdot) := \widetilde t_i^w(\theta) x_i(\theta, \cdot) + \widetilde t_i^l(\theta) (1-x_i(\theta, \cdot))
\]
satisfying the following properties:
\begin{align}
&(i) \, \widetilde T_i(\theta) \geq \bar T_i(\theta) \label{eq:FI_tilde_T} \\
&(ii) \, \widetilde U_i^{\min}(\theta)=\bar U_i^{\min}(\theta) \label{eq:FI_tilde_U} \\
&(iii) \, \alpha \widetilde t_i^w(\theta) - \beta \widetilde t_i^l(\theta) \leq K \label{eq:FI_tilde_constraint}\\
&(iv) \, \theta - \widetilde t_i^w(\theta) = -\widetilde t_i^l(\theta), \label{eq:FI_tilde_w>=l}
\end{align}
where $\widetilde U_i^{\min}(\theta)$ and $\widetilde T_i(\theta)$ denote the interim variables under $(x, \widetilde t)$.

\end{lemma}
\begin{proof}

Let $\widetilde t_i^w(\theta) := \theta - \bar U_i^{\min}(\theta)$ and $\widetilde t_i^l(\theta) := -\bar U_i^{\min}(\theta)$. A similar argument as in Case C in the proof of Corollary \ref{cor:SOSD_implication} yields properties \eqref{eq:FI_tilde_T}, \eqref{eq:FI_tilde_U} and \eqref{eq:FI_tilde_w>=l}. To verify the remaining property \eqref{eq:FI_tilde_constraint}, note that
\[
\theta-\widetilde t_i^w(\theta) = \widetilde U_i^{\min}(\theta) = \bar U_i^{\min}(\theta) \geq \theta-\bar t_i^w(\theta),
\]
where the first equality follows from $\theta-\widetilde t_i^w(\theta)=-\widetilde t_i^l(\theta)$, the second equality from property \eqref{eq:FI_tilde_U}, and the third inequality because $\theta-\bar t_i^w(\theta)<-\bar t_i^l(\theta)$. It follows that $\widetilde t_i^w(\theta) \leq \bar t_i^w(\theta)$. A similar argument yields $\widetilde t_i^l(\theta) \geq \bar t_i^l(\theta)$. Hence,
\[
\alpha \widetilde t_i^w(\theta) - \beta \widetilde t_i^l(\theta) \leq \alpha \bar t_i^w(\theta) - \beta \bar t_i^l(\theta) \leq K.
\]

\end{proof}

\begin{lemma} \label{lem:belief_minimizing_winning_prob}

Suppose that Assumption \ref{assum:div} holds. Given a Borel set $E$, let
\begin{equation} \label{eq:belief_minimizing_winning_prob}
p^{\min} := \inf_Q \{ Q(E): D(Q||P) \leq \eta \}.
\end{equation}
Define $Q^*$ as $\frac{dQ^*}{dP} = \frac{p^{\min}}{p} \cdot \mathbf{1}_E + \frac{1-p^{\min}}{1-p} \cdot \mathbf{1}_{\Theta \setminus E}$, where $p := P(E)$. Then, $Q^*$ is a minimizer of problem \eqref{eq:belief_minimizing_winning_prob}.

\end{lemma}
\begin{proof}

Plugging $\pi \equiv \mathbf{1}_E$ into equation \eqref{eq:SOSD_proof_key} in the proof of Proposition \ref{prop:SOSD} yields
\begin{equation} \label{eq:belief_minimizing_winning_prob_coarse}
\inf_Q \{ Q(E): D(Q||P) \leq \eta \} = \inf_Q \{ Q(E): \text{$D(Q||P) \leq \eta$, \, $\frac{dQ}{dP}$ is $\mathcal{E}$-measurable} \},
\end{equation}
where $\mathcal{E}$ denotes the $\sigma$-algebra generated by $E$. For $q \in [0, 1]$, there exists a unique $Q_q \in \Delta(\Theta, \mathcal{B})$ such that $\frac{dQ_q}{dP}$ is $\mathcal{E}$-measurable and $Q_q(E)=q$, defined as $\frac{dQ_q}{dP} := \frac{q}{p} \cdot \mathbf{1}_E + \frac{1-q}{1-p} \cdot \mathbf{1}_{\Theta \setminus E}$.
Hence, equation \eqref{eq:belief_minimizing_winning_prob_coarse} can be rewritten as
\begin{equation} \label{eq:belief_minimizing_winning_prob_one_dim}
\inf_Q \{ Q(E): D(Q||P) \leq \eta \} = \inf_{q \in [0, 1]} \{ q: D(Q_q||P) \leq \eta \}.
\end{equation}
By \textbf{D2} of Assumption \ref{assum:div}, the right-hand side of equation \eqref{eq:belief_minimizing_winning_prob_one_dim} has a minimizer, which must be $p^{\min}$. Thus, $Q^*=Q_{p^{\min}}$ is a minimizer to problem \eqref{eq:belief_minimizing_winning_prob}.

\end{proof}

\bibliographystyle{elsarticle-harv}
\bibliography{bib_draft}

\begin{thebibliography}{48}
\expandafter\ifx\csname natexlab\endcsname\relax\def\natexlab#1{#1}\fi
\providecommand{\url}[1]{\texttt{#1}}
\providecommand{\href}[2]{#2}
\providecommand{\path}[1]{#1}
\providecommand{\DOIprefix}{doi:}
\providecommand{\ArXivprefix}{arXiv:}
\providecommand{\URLprefix}{URL: }
\providecommand{\Pubmedprefix}{pmid:}
\providecommand{\doi}[1]{\href{http://dx.doi.org/#1}{\path{#1}}}
\providecommand{\Pubmed}[1]{\href{pmid:#1}{\path{#1}}}
\providecommand{\bibinfo}[2]{#2}
\ifx\xfnm\relax \def\xfnm[#1]{\unskip,\space#1}\fi
\bibitem[{Ali and Silvey(1966)}]{Ali66}
\bibinfo{author}{Ali, S.M.}, \bibinfo{author}{Silvey, S.D.},
  \bibinfo{year}{1966}.
\newblock \bibinfo{title}{A general class of coefficients of divergence of one
  distribution from another}.
\newblock \bibinfo{journal}{Journal of the Royal Statistical Society}
  \bibinfo{volume}{28}, \bibinfo{pages}{131--142}.
\bibitem[{Auster(2018)}]{Aus18}
\bibinfo{author}{Auster, S.}, \bibinfo{year}{2018}.
\newblock \bibinfo{title}{Robust contracting under common value uncertainty}.
\newblock \bibinfo{journal}{Theoretical Economics} \bibinfo{volume}{13},
  \bibinfo{pages}{175--204}.
\bibitem[{Auster and Kellner(2020)}]{Aus20}
\bibinfo{author}{Auster, S.}, \bibinfo{author}{Kellner, C.},
  \bibinfo{year}{2020}.
\newblock \bibinfo{title}{Robust bidding and revenue in descending price
  auctions}.
\newblock \bibinfo{note}{Forthcoming}.
\bibitem[{Ben-Tal et~al.(2013)Ben-Tal, den Hertog, Waegenaere, Melenberg and
  Rennen}]{Ben13}
\bibinfo{author}{Ben-Tal, A.}, \bibinfo{author}{den Hertog, D.},
  \bibinfo{author}{Waegenaere, A.D.}, \bibinfo{author}{Melenberg, B.},
  \bibinfo{author}{Rennen, G.}, \bibinfo{year}{2013}.
\newblock \bibinfo{title}{Robust solutions of optimization problems affected by
  uncertain probabilities}.
\newblock \bibinfo{journal}{Management Science} \bibinfo{volume}{59},
  \bibinfo{pages}{341--357}.
\bibitem[{Bewley(2002)}]{Bew02}
\bibinfo{author}{Bewley, T.F.}, \bibinfo{year}{2002}.
\newblock \bibinfo{title}{Knightian decision theory. {P}art {I}}.
\newblock \bibinfo{journal}{Decisions in Economics and Finance}
  \bibinfo{volume}{25}, \bibinfo{pages}{79--110}.
\bibitem[{Bodoh-Creed(2012)}]{Bodoh12}
\bibinfo{author}{Bodoh-Creed, A.L.}, \bibinfo{year}{2012}.
\newblock \bibinfo{title}{Ambiguous beliefs and mechanism design}.
\newblock \bibinfo{journal}{Games and Economic Behavior} \bibinfo{volume}{75},
  \bibinfo{pages}{518--537}.
\bibitem[{Bodoh-Creed(2014)}]{Bodoh14}
\bibinfo{author}{Bodoh-Creed, A.L.}, \bibinfo{year}{2014}.
\newblock \bibinfo{title}{Correction to ambiguous beliefs and mechanism
  design}.
\bibitem[{Bos(2012)}]{Bos2012}
\bibinfo{author}{Bos, O.}, \bibinfo{year}{2012}.
\newblock \bibinfo{title}{Wars of attrition and all-pay auctions with
  stochastic competition}.
\newblock \bibinfo{journal}{Journal of Mathematical Economics}
  \bibinfo{volume}{48}, \bibinfo{pages}{83--91}.
\bibitem[{Bose and Daripa(2009)}]{Bo09}
\bibinfo{author}{Bose, S.}, \bibinfo{author}{Daripa, A.}, \bibinfo{year}{2009}.
\newblock \bibinfo{title}{A dynamic mechanism and surplus extraction under
  ambiguity}.
\newblock \bibinfo{journal}{Journal of Economic Theory} \bibinfo{volume}{144},
  \bibinfo{pages}{2084--2114}.
\bibitem[{Bose et~al.(2006)Bose, Ozdenoren and Pape}]{Bo06}
\bibinfo{author}{Bose, S.}, \bibinfo{author}{Ozdenoren, E.},
  \bibinfo{author}{Pape, A.}, \bibinfo{year}{2006}.
\newblock \bibinfo{title}{Optimal auctions with ambiguity}.
\newblock \bibinfo{journal}{Theoretical Economics} \bibinfo{volume}{1},
  \bibinfo{pages}{411--438}.
\bibitem[{Brunner et~al.(2014)Brunner, Hu and Oechssler}]{Brun14}
\bibinfo{author}{Brunner, C.}, \bibinfo{author}{Hu, A.},
  \bibinfo{author}{Oechssler, J.}, \bibinfo{year}{2014}.
\newblock \bibinfo{title}{Premium auctions and risk preferences: {A}n
  experimental study}.
\newblock \bibinfo{journal}{Games and Economic Behavior} \bibinfo{volume}{87},
  \bibinfo{pages}{467--484}.
\bibitem[{Chiesa et~al.(2015)Chiesa, Micali and Zhu}]{Chi15}
\bibinfo{author}{Chiesa, A.}, \bibinfo{author}{Micali, S.},
  \bibinfo{author}{Zhu, Z.A.}, \bibinfo{year}{2015}.
\newblock \bibinfo{title}{Knightian analysis of the {V}ickrey auction}.
\newblock \bibinfo{journal}{Econometrica} \bibinfo{volume}{83},
  \bibinfo{pages}{1727--1754}.
\bibitem[{Cohen et~al.(2008)Cohen, Kaplan and Sela}]{Coh08}
\bibinfo{author}{Cohen, C.}, \bibinfo{author}{Kaplan, T.R.},
  \bibinfo{author}{Sela, A.}, \bibinfo{year}{2008}.
\newblock \bibinfo{title}{Optimal rewards in contests}.
\newblock \bibinfo{journal}{RAND Journal of Economics} \bibinfo{volume}{39},
  \bibinfo{pages}{434--451}.
\bibitem[{Cox and Hinkley(1974)}]{Cox74}
\bibinfo{author}{Cox, D.R.}, \bibinfo{author}{Hinkley, D.V.},
  \bibinfo{year}{1974}.
\newblock \bibinfo{title}{Theoretical statistics}.
\newblock \bibinfo{publisher}{Chapman and Hall}, \bibinfo{address}{London}.
\bibitem[{Duchi et~al.(2021)Duchi, Glynn and Namkoong}]{Duchi21}
\bibinfo{author}{Duchi, J.C.}, \bibinfo{author}{Glynn, P.W.},
  \bibinfo{author}{Namkoong, H.}, \bibinfo{year}{2021}.
\newblock \bibinfo{title}{Statistics of robust optimization: {A} generalized
  empirical likelihood approach}.
\newblock \bibinfo{note}{Forthcoming}.
\bibitem[{Dupuis et~al.(2020)Dupuis, Katsoulakis, Pantazis and
  Rey-Bellet}]{Dupuis20}
\bibinfo{author}{Dupuis, P.}, \bibinfo{author}{Katsoulakis, M.A.},
  \bibinfo{author}{Pantazis, Y.}, \bibinfo{author}{Rey-Bellet, L.},
  \bibinfo{year}{2020}.
\newblock \bibinfo{title}{Sensitivity analysis for rare events based on
  {R}\'enyi divergence}.
\newblock \bibinfo{journal}{The Annals of Applied Probability}
  \bibinfo{volume}{30}, \bibinfo{pages}{1507--1533}.
\bibitem[{Gentzkow and Kamenica(2014)}]{Gent14}
\bibinfo{author}{Gentzkow, M.}, \bibinfo{author}{Kamenica, E.},
  \bibinfo{year}{2014}.
\newblock \bibinfo{title}{Costly persuasion}.
\newblock \bibinfo{journal}{American Economic Review} \bibinfo{volume}{104},
  \bibinfo{pages}{457--462}.
\bibitem[{Ghosh and Liu(2021)}]{Ghosh21}
\bibinfo{author}{Ghosh, G.}, \bibinfo{author}{Liu, H.}, \bibinfo{year}{2021}.
\newblock \bibinfo{title}{Sequential auctions with ambiguity}.
\newblock \bibinfo{journal}{Journal of Economic Theory} \bibinfo{volume}{197},
  \bibinfo{pages}{105324}.
\bibitem[{Gilboa and Schmeidler(1989)}]{Gil89}
\bibinfo{author}{Gilboa, I.}, \bibinfo{author}{Schmeidler, D.},
  \bibinfo{year}{1989}.
\newblock \bibinfo{title}{Maxmin expected utility with non-unique prior}.
\newblock \bibinfo{journal}{Journal of Mathematical Economics}
  \bibinfo{volume}{18}, \bibinfo{pages}{141--153}.
\bibitem[{Goeree and Offerman(2004)}]{Goe04}
\bibinfo{author}{Goeree, J.K.}, \bibinfo{author}{Offerman, T.},
  \bibinfo{year}{2004}.
\newblock \bibinfo{title}{The {A}msterdam auction}.
\newblock \bibinfo{journal}{Econometrica} \bibinfo{volume}{72},
  \bibinfo{pages}{281--294}.
\bibitem[{Gourgoulias et~al.(2020)Gourgoulias, Katsoulakis, Rey-Bellet and
  Wang}]{Gour20}
\bibinfo{author}{Gourgoulias, K.}, \bibinfo{author}{Katsoulakis, M.A.},
  \bibinfo{author}{Rey-Bellet, L.}, \bibinfo{author}{Wang, J.},
  \bibinfo{year}{2020}.
\newblock \bibinfo{title}{How biased is your model? {C}oncentration
  inequalities, information and model bias}.
\newblock \bibinfo{journal}{IEEE Transactions on Information Theory}
  \bibinfo{volume}{66}, \bibinfo{pages}{3079--3097}.
\bibitem[{Hansen and Sargent(2001)}]{Han01}
\bibinfo{author}{Hansen, L.P.}, \bibinfo{author}{Sargent, T.},
  \bibinfo{year}{2001}.
\newblock \bibinfo{title}{Robust control and model uncertainty}.
\newblock \bibinfo{journal}{American Economic Review} \bibinfo{volume}{91},
  \bibinfo{pages}{60--66}.
\bibitem[{Hansen and Sargent(2008)}]{Han08}
\bibinfo{author}{Hansen, L.P.}, \bibinfo{author}{Sargent, T.},
  \bibinfo{year}{2008}.
\newblock \bibinfo{title}{Robustness}.
\newblock \bibinfo{publisher}{Princeton University Press}.
\bibitem[{H{\"o}risch(2010)}]{Horisch10}
\bibinfo{author}{H{\"o}risch, Hannah, O.K.}, \bibinfo{year}{2010}.
\newblock \bibinfo{title}{Less fighting than expected: Experiments with wars of
  attrition and all-pay auctions}.
\newblock \bibinfo{journal}{Public Choice} \bibinfo{volume}{144},
  \bibinfo{pages}{347--367}.
\bibitem[{Hu et~al.(2011)Hu, Offerman and Zou}]{Hu11}
\bibinfo{author}{Hu, A.}, \bibinfo{author}{Offerman, T.}, \bibinfo{author}{Zou,
  L.}, \bibinfo{year}{2011}.
\newblock \bibinfo{title}{Premium auctions and risk preferences}.
\newblock \bibinfo{journal}{Journal of Economic Theory} \bibinfo{volume}{146},
  \bibinfo{pages}{2420--2439}.
\bibitem[{Hu et~al.(2018)Hu, Offerman and Zou}]{Hu18}
\bibinfo{author}{Hu, A.}, \bibinfo{author}{Offerman, T.}, \bibinfo{author}{Zou,
  L.}, \bibinfo{year}{2018}.
\newblock \bibinfo{title}{How risk sharing may enhance efficiency of english
  auctions}.
\newblock \bibinfo{journal}{The Economic Journal} \bibinfo{volume}{128},
  \bibinfo{pages}{1235--1256}.
\bibitem[{Kaplan et~al.(2002)Kaplan, Luski, Sela and Wettstein}]{Kap02}
\bibinfo{author}{Kaplan, T.}, \bibinfo{author}{Luski, I.},
  \bibinfo{author}{Sela, A.}, \bibinfo{author}{Wettstein, D.},
  \bibinfo{year}{2002}.
\newblock \bibinfo{title}{All-pay auctions with variable rewards}.
\newblock \bibinfo{journal}{The Journal of Industrial Economics}
  \bibinfo{volume}{50}, \bibinfo{pages}{417--430}.
\bibitem[{Kocherlakota and Song(2019)}]{Kocher19}
\bibinfo{author}{Kocherlakota, N.R.}, \bibinfo{author}{Song, Y.},
  \bibinfo{year}{2019}.
\newblock \bibinfo{title}{Public goods with ambiguity in large economies}.
\newblock \bibinfo{journal}{Journal of Economic Theory} \bibinfo{volume}{182},
  \bibinfo{pages}{218--246}.
\bibitem[{Kocyigit et~al.(2020)Kocyigit, Iyengar, Kuhn and Wiesemann}]{Kocyi20}
\bibinfo{author}{Kocyigit, C.}, \bibinfo{author}{Iyengar, G.},
  \bibinfo{author}{Kuhn, D.}, \bibinfo{author}{Wiesemann, W.},
  \bibinfo{year}{2020}.
\newblock \bibinfo{title}{Distributionally robust mechanism design}.
\newblock \bibinfo{journal}{Management Science} \bibinfo{volume}{66},
  \bibinfo{pages}{159--189}.
\bibitem[{Krishna and Morgan(1997)}]{Kris97}
\bibinfo{author}{Krishna, V.}, \bibinfo{author}{Morgan, J.},
  \bibinfo{year}{1997}.
\newblock \bibinfo{title}{An analysis of the war of attrition and the all-pay
  auction}.
\newblock \bibinfo{journal}{Journal of Economic Theory} \bibinfo{volume}{72},
  \bibinfo{pages}{343--362}.
\bibitem[{Kullback(1959)}]{Kull59}
\bibinfo{author}{Kullback, S.}, \bibinfo{year}{1959}.
\newblock \bibinfo{title}{Information theory and statistics}.
\newblock \bibinfo{publisher}{Wiley}, \bibinfo{address}{New York}.
\bibitem[{Kullback and Leibler(1951)}]{Kull51}
\bibinfo{author}{Kullback, S.}, \bibinfo{author}{Leibler, R.A.},
  \bibinfo{year}{1951}.
\newblock \bibinfo{title}{On information and sufficiency}.
\newblock \bibinfo{journal}{The Annals of Mathematical Statistics}
  \bibinfo{volume}{22}, \bibinfo{pages}{79--86}.
\bibitem[{Laohakunakorn et~al.(2019)Laohakunakorn, Levy and Razin}]{Lao19}
\bibinfo{author}{Laohakunakorn, K.}, \bibinfo{author}{Levy, G.},
  \bibinfo{author}{Razin, R.}, \bibinfo{year}{2019}.
\newblock \bibinfo{title}{Private and common value auctions with ambiguity over
  correlation}.
\newblock \bibinfo{journal}{Journal of Economic Theory} \bibinfo{volume}{184},
  \bibinfo{pages}{104932}.
\bibitem[{Lo(1998)}]{Lo98}
\bibinfo{author}{Lo, K.C.}, \bibinfo{year}{1998}.
\newblock \bibinfo{title}{Sealed bid auctions with uncertainty averse bidders}.
\newblock \bibinfo{journal}{Economic Theory} \bibinfo{volume}{12},
  \bibinfo{pages}{1--20}.
\bibitem[{Matejka and McKay(2015)}]{Mat15}
\bibinfo{author}{Matejka, F.}, \bibinfo{author}{McKay, A.},
  \bibinfo{year}{2015}.
\newblock \bibinfo{title}{Rational inattention to discrete choices: A new
  foundation for the multinomial logit model}.
\newblock \bibinfo{journal}{American Economic Review} \bibinfo{volume}{105},
  \bibinfo{pages}{272--298}.
\bibitem[{Matros(2012)}]{Matros12}
\bibinfo{author}{Matros, A.}, \bibinfo{year}{2012}.
\newblock \bibinfo{title}{Sad-loser contests}.
\newblock \bibinfo{journal}{Journal of Mathematical Economics}
  \bibinfo{volume}{48}, \bibinfo{pages}{155--162}.
\bibitem[{Matros and Armanios(2009)}]{Matros09}
\bibinfo{author}{Matros, A.}, \bibinfo{author}{Armanios, D.},
  \bibinfo{year}{2009}.
\newblock \bibinfo{title}{Tullock's contest with reimbursements}.
\newblock \bibinfo{journal}{Public Choice} \bibinfo{volume}{141},
  \bibinfo{pages}{49--63}.
\bibitem[{Milgrom(2004)}]{Mil04}
\bibinfo{author}{Milgrom, P.}, \bibinfo{year}{2004}.
\newblock \bibinfo{title}{Putting auction theory to work}.
\newblock \bibinfo{publisher}{Cambridge University Press}.
\bibitem[{Milgrom and Weber(1982)}]{Mil82}
\bibinfo{author}{Milgrom, P.}, \bibinfo{author}{Weber, R.},
  \bibinfo{year}{1982}.
\newblock \bibinfo{title}{A theory of auctions and competitive bidding}.
\newblock \bibinfo{journal}{Econometrica} \bibinfo{volume}{50},
  \bibinfo{pages}{1089--1122}.
\bibitem[{Minchuk(2018)}]{Minchuk18}
\bibinfo{author}{Minchuk, Y.}, \bibinfo{year}{2018}.
\newblock \bibinfo{title}{Effect of reimbursement on all-pay auction}.
\newblock \bibinfo{journal}{Economics Letters} \bibinfo{volume}{172},
  \bibinfo{pages}{28--30}.
\bibitem[{Minchuk and Sela(2020)}]{Min20}
\bibinfo{author}{Minchuk, Y.}, \bibinfo{author}{Sela, A.},
  \bibinfo{year}{2020}.
\newblock \bibinfo{title}{Contests with insurance}.
\newblock \bibinfo{journal}{Review of Economic Design} \bibinfo{volume}{24},
  \bibinfo{pages}{1--22}.
\bibitem[{Myerson(1981)}]{Myer81}
\bibinfo{author}{Myerson, R.B.}, \bibinfo{year}{1981}.
\newblock \bibinfo{title}{Optimal auction design}.
\newblock \bibinfo{journal}{Mathematics of Operations Research}
  \bibinfo{volume}{6}, \bibinfo{pages}{58--73}.
\bibitem[{Riley and Samuelson(1981)}]{Ril81}
\bibinfo{author}{Riley, J.}, \bibinfo{author}{Samuelson, W.},
  \bibinfo{year}{1981}.
\newblock \bibinfo{title}{Optimal auctions}.
\newblock \bibinfo{journal}{American Economic Review} \bibinfo{volume}{71},
  \bibinfo{pages}{381--392}.
\bibitem[{Sims(2003)}]{Sims03}
\bibinfo{author}{Sims, C.A.}, \bibinfo{year}{2003}.
\newblock \bibinfo{title}{Implications of rational inattention}.
\newblock \bibinfo{journal}{Journal of Monetary Economics}
  \bibinfo{volume}{50}, \bibinfo{pages}{665--690}.
\bibitem[{Stong(2018)}]{Stong18}
\bibinfo{author}{Stong, S.}, \bibinfo{year}{2018}.
\newblock \bibinfo{title}{Ambiguity aversion in the all-pay auction and war of
  attrition}.
\newblock \bibinfo{journal}{Journal of Public Economic Theory}
  \bibinfo{volume}{20}, \bibinfo{pages}{822--839}.
\bibitem[{Treust and Tomala(2019)}]{Treu19}
\bibinfo{author}{Treust, M.L.}, \bibinfo{author}{Tomala, T.},
  \bibinfo{year}{2019}.
\newblock \bibinfo{title}{Persuasion with limited communication capacity}.
\newblock \bibinfo{journal}{Journal of Economic Theory} \bibinfo{volume}{184},
  \bibinfo{pages}{104940}.
\bibitem[{Wilson(1987)}]{Wil87}
\bibinfo{author}{Wilson, R.}, \bibinfo{year}{1987}.
\newblock \bibinfo{title}{Game-theoretic analyses of trading processes}, in:
  \bibinfo{editor}{Bewley, T.} (Ed.), \bibinfo{booktitle}{Advances in Economic
  Theory: Fifth World Congress}. \bibinfo{publisher}{Cambridge University
  Press}, \bibinfo{address}{Cambridge, U. K.}. chapter~\bibinfo{chapter}{2},
  pp. \bibinfo{pages}{33--70}.
\bibitem[{Wolitzky(2016)}]{Wol16}
\bibinfo{author}{Wolitzky, A.}, \bibinfo{year}{2016}.
\newblock \bibinfo{title}{Mechanism design with maxmin agents: Theory and an
  application to bilateral trade}.
\newblock \bibinfo{journal}{Theoretical Economics} \bibinfo{volume}{11},
  \bibinfo{pages}{971--1004}.

\end{thebibliography}

\end{document}